\documentclass[journal]{IEEEtran}

\ifCLASSINFOpdf
\else
\fi

\usepackage{url}
\usepackage{subcaption}
\usepackage{graphicx}
\usepackage{xcolor}
\usepackage{soul}  
\usepackage{amsmath}
\usepackage{amssymb}
\usepackage{amsfonts}
\usepackage[nice]{nicefrac}
\usepackage{amsthm}
\usepackage{algorithm}
\usepackage{algorithmic}

\usepackage{ulem}
\usepackage{latexsym,mathrsfs}
\usepackage{enumerate,verbatim}

\usepackage{mathtools}

	\definecolor{brightpink}{rgb}{1.0, 0.0, 0.5}

	\newcommand{\revise}[1]{{{\color{black} #1}}}

\newtheorem{lemma}{Lemma}

\newtheorem{theorem}{Theorem}

\newtheorem{definition}{Definition}
\newtheorem*{definition*}{Definition}

\begin{document}

\title{Multi-Resolution Beta-Divergence NMF for \\ Blind Spectral Unmixing}

\author{Valentin~Leplat,
        Nicolas~Gillis,
        and~C\'edric~F\'evotte
\thanks{Valentin Leplat is with Center for Artificial Intelligence Technology, Skoltech, 
Bolshoy Boulevard 30, bld. 1. Moscow, Russia. He acknowledges the support by the European Research Council (ERC Advanced Grant no 788368) and the support by Ministry of Science and Higher Education grant No. 075-10-2021-068.
Email: V.Leplat@skoltech.ru.}
\thanks{Nicolas Gillis is with Department of Mathematics and Operational Research,
Facult\'e Polytechnique, Universit\'e de Mons,
Rue de Houdain 9, 7000 Mons, Belgium. NG acknowledges the support by the Fonds de la Recherche Scientifique - FNRS and the Fonds Wetenschappelijk Onderzoek - Vlanderen (FWO) under EOS Project no O005318F-RG47, by the European Research Council (ERC Starting Grant no 679515), and by the Francqui foundation. 
E-mails: nicolas.gillis@umons.ac.be.}
\thanks{C\'edric~F\'evotte is with IRIT, Université de Toulouse, CNRS, Toulouse, France. He acknowledges the support by the European Research Council (ERC Consolidator Grant no 6681839).
E-mail: cedric.fevotte@irit.fr.}
\thanks{Manuscript received in June 2021.}}

\markboth{Submitted to Signal Processing}{Leplat \MakeLowercase{\textit{et al.}}: Multi-resolution $\beta$ NMF}

\maketitle

\begin{abstract}

Many datasets are obtained as a resolution trade-off between two adversarial dimensions; for example between the frequency and the temporal resolutions  for the spectrogram of an audio signal, and between the number of wavelengths and the spatial resolution for a hyper/multi-spectral image. To perform blind source separation using observations with different resolutions, a standard approach is to use coupled nonnegative matrix factorizations (NMF). {As opposed to most previous works focusing on the least squares error measure, which is the $\beta$-divergence for $\beta = 2$, we formulate this multi-resolution NMF problem for any $\beta$-divergence, and propose a novel algorithm based on the   multiplicative updates (MU).} 
We show on numerical experiments that the MU are able to obtain high resolutions in both dimensions on two applications: (1)~blind unmixing of audio spectrograms: {to the best of our knowledge, this is the first time a coupled NMF model is used in this context}, and (2)~the fusion of hyperspectral and multispectral images: we show that the MU compete favorably with state-of-the-art algorithms in particular in the presence of non-Gaussian noise.

\end{abstract}


%
\IEEEpeerreviewmaketitle


\section{Introduction}\label{sec_intro}

\IEEEPARstart{S}{pectral} unmixing is the problem of decomposing the spectra of mixed signals into a set of source spectra and their corresponding activations. The activations give the proportion of each source within each mixed signal spectrum.   
In this paper, we focus on the unmixing of nonnegative signals with multiple resolutions using nonnegative matrix factorization (NMF). 

\subsection{Nonnegative matrix factorization} Over the last two decades, NMF~\cite{lee1999learning} has emerged as a useful method to decompose 
mixed nonnegative signals, including audio signals 
\cite{6784107}
and hyperspectral images~\cite{Bioucas2012, ma2014signal}; see also~\cite{8653529, Gillis2020} 
and the references therein. 
Given a nonnegative matrix $V \in \mathbb{R}_{+}^{F \times N}$ and an integer factorization rank  $K \leq \text{min}(F,N)$, NMF aims to compute a  nonnegative matrix $W$ with $K$ columns and a  nonnegative matrix $H$ with $K$ rows such that $V \approx WH$. 
Each column of $V$ is the mixture of the sources, so that each column of $W$ corresponds to a source estimate, and each column of $H$ indicates which source is active and in which intensity in each mixture. Mathematically, we have, for all $j$, 
\[
V(:,j) \approx \sum_{k=1}^K W(:,k) H(k,j) , 
\] 
where $W(:,k)$ represents the $k$th source, and $H(k,j)$ is the activation of the $k$th source within the $j$th mixture.

\subsection{Multi-resolution data} In many applications, the input data usually results from a trade-off between two adversarial dimensions. Let us illustrate this on two applications which will be used throughout the paper.  

\paragraph{Audio signals} 
To unmix  audio signals, their time-frequency matrix representation $V$ is often used; see, e.g.,~\cite{smaragdis2003non,fevotte2009nonnegative,6784107}. 
In a nutshell, this matrix is computed as follows. 
The temporal audio signal is divided into short segments of the same length. These segments are  multiplied by a window function and then the {magnitude} Fourier transform of each windowed segment is computed to obtain a column of $V$. Hence each column of $V$ corresponds to a time window, while each row corresponds to a frequency, and the entry $V(i,j)$ is the intensity of the $i$th frequency  at the $j$th time window (e.g., the modulus of the Fourier coefficient). 
The window length fixes the frequency and the time resolutions. Larger time windows lead to a higher frequency resolution but comes at the cost of lower temporal resolution, and vice versa. 
Factorizing $V$ using NMF provides the matrix $W$ whose columns contain the spectral content of the sources, and the matrix $H$ whose rows contain the activations of the sources over time; 
see~\cite{6784107} and the references therein for more details. 

\paragraph{Hyper/multi-spectral images} An image measures the intensity of light in both spectral and spatial dimensions. 
A multispectral image (MSI)  typically measures between 4 and 30 spectral bands, and has a high spatial resolution,  whereas a hyperspectral  image (HSI) has high spectral resolution, typically between 100 and 200 spectral bands, but low spatial resolution. 
MSI/HSI are typically represented as a wavelength-by-pixel nonnegative matrix $V$ where the entry $V(i,j)$ is the intensity of light at the $i$th wavelength located at the $j$th pixel. Each column of $V$ records the so-called spectral signature of a pixel, and each row is a vectorized image at a given wavelength. 
Factorizing $V$ using NMF gives the matrix $W$ whose columns contain the spectral signatures of the sources, called endmembers, 
and the matrix $H$ whose rows contain the abundances of the pixels for each endmember; 
see~\cite{Bioucas2012, ma2014signal}  and the references therein for more details. 
Given a MSI and a HSI of the same scene, computing a high spatial and spectral resolution image of that scene, referred to as the super-resolution (SR) image, 
is known as the HSI-MSI fusion problem which has been extensively studied;  see for 
example~\cite{Price1987,
Gillespie1987, 
Carper1990,
Chavez1991,
Nishii1996,
zhukov1999unmixing,  
Ranchin2000,
Aiazzi2002, 
zurita2008unmixing, 
Yokoya2012CNMF, 
wei2015hyperspectral,
yokoya2017, lin2017convex, 
Kanatsoulis2018Ma}. 
A popular and effective method to perform this task is to perform coupled NMF decompositions of the HSI and MSI; see Section~\ref{sec_probform} for the details.


\paragraph*{Contribution and outline}

In this paper, we consider the fusion of multi-resolution data using a coupled NMF model which is described in Section~\ref{sec_probform}. 
{As opposed to most previous works focused on the case $\beta = 2$, that is, least squares error,} we allow to use any $\beta$-divergence ({a large family of divergences commonly used in NMF \cite{Fevotte_betadiv}}) to measure the quality of the low-rank approximation. 
We refer to this model as multi-resolution $\beta$-NMF (MR-$\beta$-NMF). 
To tackle MR-$\beta$-NMF, we propose in Section~\ref{sec_algo} multiplicative updates that are guaranteed to decrease the objective function at each step, using the majorization-minimization principle. 
{This principle is already present in the NMF literature, but we adapted it to handle our specific model which is more challenging.} 
We also explain how the downsampling operators (that map high resolution data to low resolution data) can be estimated for one-dimensional signals such as audio signals.  Section~\ref{sec_audioexp} presents numerical results on audio datasets: {as far as we know, it is the first time such an approach is used in this context.}  
For audio signals, it is well-known that using $\beta$-divergences for $\beta < 2$ is crucial in practice; see, e.g.,~\cite{fevotte2009nonnegative,6784107}. 
MR-$\beta$-NMF  leads to solutions with both high spectral resolution and high temporal resolution. 
In Section~\ref{sec_hsmsexp}, MR-$\beta$-NMF is shown to be competitive with state-of-the-art techniques for the HSI-MSI fusion problem. {As far as we know, it is the first time that a HSI-MSI fusion algorithm tackles $\beta$-divergence for $\beta \neq 2$.}  
As we will see, considering $\beta$-divergences for $\beta \neq 2$ leads to much better solutions in the presence of non-Gaussian noise. In particular, we show that in the presence of Poisson noise, using $\beta = 1$, that is, the Kullback-Leibler divergence, outperforms standard approaches. {In the appendix, we also show that in the presence of multiplicative Gamma noise, our proposed  MR-$\beta$-NMF with the Itakura-Saito divergence ($\beta = 0$) outperforms the state of the art by a large margin.}


\section{Formulation of MR-$\beta$-NMF}\label{sec_probform}

The aim of multi-resolution unmixing is to estimate the sources and their 
activation with high resolutions in adversarial dimensions, given  observable data that show high resolution in one dimension only. 

In this section, we present a model widely used in the hyperspectral imaging community for HSI-MSI fusion. As we will see in Section~\ref{sec_numexpaudio}, this model is also applicable to decompose audio signals. 
For simplicity, we assume in this paper that we are given only two input data matrices, one with low resolution in one dimension, and the other with low resolution in the other dimension. 
Generalizing to more than two input data matrices with different resolutions in the two dimensions is straightforward but would complicate the presentation. 

Let $X \in \mathbb{R}_{+}^{F_\ell \times N}$ and $Y \in \mathbb{R}_{+}^{F \times N_\ell}$ be these two matrices, 
where $X$ has low resolution in the first dimension, that is, $F_{\ell} < F$, and 
$Y$ has low resolution in the second dimension, that is, $N_{\ell} < N$. 
Given $X$ and $Y$, 
the goal is twofold: 
(1)~compute $V \in \mathbb{R}_{+}^{F \times N}$ that has  high  resolution in both dimensions, and 
(2)~identify the sources and activations that generated $X$ and  $Y$. 
{A standard approach to achieve these  
goals~\cite{Yokoya2012CNMF, 
wycoff2013non, 
wei2015fast, wei2015hyperspectral, 
xu2015spatial, 
zhang2015hyperspectral, 
zhou2017hyperspectral, lin2017convex} 
is to rely on the following two assumptions:  }
\begin{enumerate}
    \item {The matrix $V$ satisfies the linear mixing model, that is, $V$ can be decomposed using NMF with 
    \begin{equation}\label{approxV}
V \approx WH, 
\end{equation} 
    where the columns of $W \in \mathbb{R}_{+}^{F \times K}$ are the {elementary spectra of the sources},  
     $H \in \mathbb{R}_{+}^{K \times N}$ is the activation matrix, and $K$ is the number of sources that generated $V$ (e.g., the number of endmembers in a HSI);   see~\cite{Bioucas2012, ma2014signal} and the references therein in the context of HSIs, and \cite{fevotte2009nonnegative} in the context of audio signals. } 
    
    \item The matrices $X$ and $Y$ are obtained using linear downsampling operators of $V$, that is, 
    \begin{equation}\label{approxV1}
X \approx RV, 
\end{equation} 
where $R \in \mathbb{R}^{F_{\ell} \times F}$ is the downsampling matrix in the first dimension,  and 
\begin{equation}\label{approxV2}
Y \approx VS, 
\end{equation} 
where $S \in \mathbb{R}^{N \times N_{\ell}}$ is the downsampling matrix in the second dimension. 
In practice, these downsampling matrices need to be estimated. 
\end{enumerate}

For HSI-MSI fusion, a high spatial resolution image $X$, the MSI, and a high spectral resolution image $Y$, the HSI, are available to reconstruct 
the target SR image, $V$, that has high spectral and high spatial resolutions. 
These images result from the linear spectral and spatial degradations of the SR image $V$, given by the equations  \eqref{approxV1} and \eqref{approxV2}.

\subsection{Multi-Resolution $\beta$-NMF}\label{sec2_2}

Substituting~\eqref{approxV} into~\eqref{approxV1} 
and~\eqref{approxV2}, we obtain:  
\begin{equation}\label{approxV1_final}
X \approx RWH,
\end{equation} 
\begin{equation}\label{approxV2_final}
Y \approx WHS. 
\end{equation} 
Equation~\eqref{approxV1_final} (resp.~\eqref{approxV2_final}) correspond to the linear spectral mixture model degraded in the first (resp.\ second) dimension.  

Given $X$ and $Y$, to solve the multi-resolution problem and obtain $V$, 
we need to estimate $W$, $H$, $R$ and $S$. 
Trying to minimize the approximation errors 
in~\eqref{approxV1_final} and \eqref{approxV2_final} leads to the following optimization problem 
\begin{equation}\label{MR_betaNMF_model}
 \min_{W \geq 0, H\geq 0, R \geq 0, S \geq 0} 
 D_{\beta}(X \| RWH) + \lambda D_{\beta}(Y\|WHS), 
\end{equation} 
where 
$A \geq 0$ means that $A$ is component-wise nonnegative, $\lambda$ is a positive penalty parameter, and 
\[
D_{\beta}(Z\|ABC)=\sum_{fn} d(Z_{fn}\|[ABC]_{fn}), 
\] 
with $d(x\|y)$ a {measure of fit} 
between the scalars $x$ and $y$. 
This model is the coupled NMF approach proposed \revise{in the literature for 
HSI-MSI 
fusion~\cite{yokoya2011coupled,Yokoya2012CNMF,zhang2016fusion,yokoya2017,wu2018hi,li2020sparsity}.}


When the downsampling matrices $R$ and $S$ are known, the objective function is minimized over $W$ and $H$ only. 
In general $R$ and $S$ respect a particular sparsity pattern; for example, for HSI-MSI fusion, the spectral signature of a pixel in $X$ will be a linear combination of the spectral signatures of nearby pixels from $V$; see Section~\ref{sec:downsampop} for more details. 
As our algorithm will rely on multiplicative updates, entries initialized at zero remain zero in the course of the optimization process. 

One of the most widely used measure of fit in the NMF literature is the  $\beta$-divergence, denoted $d_{\beta}(x\|y)$, and equal to 
\[
\left\{
  \begin{array}{lr}
    \frac{1}{\beta \left(\beta-1\right)}   \left(x^{\beta}+\left(\beta-1\right)y^{\beta}-\beta xy^{\beta-1}\right)  \text{ for } 
     \beta \neq 0,1, \\
    x\log\frac{x}{y}-x+y           
         \text{ for } \beta=1,  \\
    \frac{x}{y}-\log\frac{x}{y}-1  
        \text{ for } \beta=0,  
  \end{array}
\right. 
\] 
where $x$ and $y$ are nonnegative scalars. 
For $\beta=2$, this amounts to the standard squared Euclidean distance since $d_{2}(x\|y) = \nicefrac{1}{2} (x-y)^2$. 
For $\beta=1$ and $\beta=0$, 
the $\beta$-divergence corresponds to the Kullback-Leibler (KL) divergence and the Itakura-Saito (IS) divergence, respectively. 
The data fitting term should be chosen depending on the noise statistic assumed on the data. For example, 
using the Euclidean distance corresponds to the maximum likelihood estimator for i.i.d.\@ Gaussian noise, that is, it assumes that $V(i,j)$ is a sample of the normal distribution of mean $(WH)_{i,j}$ and variance $\sigma$ for all $i,j$. Similarly, the 
KL divergence corresponds to a Poisson distribution, and the IS divergence to multiplicative Gamma noise; see~\cite{fevotte2009nonnegative} for more details. 
KL and IS divergences are usually considered for amplitude spectrogram and power spectrogram, respectively.
Both KL and IS divergences are more adapted to audio spectral unmixing than the Euclidean distance; see \cite{fevotte2009nonnegative,mlsp12}. The \revise{Euclidean} distance is the most widely used to tackle the HSI unmixing problem as well as the HSI-MSI fusion problem. 
However, when no obvious choice of a specific divergence is available, finding the right measure of fit, namely the value for $\beta$,  is a model selection problem~\cite{7194802}. 


\subsection{Downsampling matrices} \label{sec:downsampop}

The downsampling matrices, 
$R$ and $S$ in~\eqref{approxV1}, are application dependent. Let us discuss the two applications we focus on in this paper.

\subsubsection{HSI-MSI fusion}

The matrix $R$ from~\eqref{approxV1} is the relative spectral bandpass responses from the SR image to the MSI, while the matrix $S$ introduced in \eqref{approxV2} specifies the spatial blurring and down-sampling responses that
result in the HSI. The matrices $R$ and $S$ can be acquired either by cross-calibration~\cite{yokoya2013}, or by estimations 
from the HSI and MSI~\cite{Yokoya2012CNMF,Simoes2016}.

\subsubsection{Audio spectral unmixing}

In the case of the audio spectral unmixing, as we restrict to the case of two input matrices, the unmixing will be based on a  
high-frequency-resolution (HFR) matrix and low-frequency-resolution (LFR) matrix, the first one obtained with a smaller window size when computing the spectrogram. 
As far as we know, there is no prior work on estimating the downsampling matrices, $R$ and $S$, as the fusion problem is considered for the first time in this paper. 
We have tested different structures for downsampling matrices $R$ and $S$, and we report here the form for $R$ that shows the best results in practice, while $S$ is obtained in the same way. This structure is a simple one-dimensional downsampling  linear operator, but turns out to perform well in practice. 
Let us illustrate this on the simple example of the frequency downsampling of a matrix $W \in \mathbb{R}^{8 \times 3}$  with a downsampling ratio $d=2$. 
A possible structure  for the matrix $R \in \mathbb{R}^{4 \times 8}_+$  is as follows:
    \begin{equation*} \label{form3_gen}
	R = \left( 
	\begin{array}{cccccccc}
	\underline{r_{11}} &  \underline{r_{12}} & \mathbf{r_{13}}&    0  & 0 & 0 & 0 & 0 \\
	0      &       \mathbf{r_{21}} &   \underline{r_{22}} &  \underline{r_{23}}  & \mathbf{r_{24}} & 0 & 0 & 0 \\
	0      &       0 &              0 &    \mathbf{r_{31}} &   \underline{r_{32}} &  \underline{r_{33}}  & \mathbf{r_{34}} & 0 \\
	0      &       0 &              0 &    0  & 0 & \mathbf{r_{41}} & \underline{r_{42}} &  \underline{r_{43}}
	\end{array} 
	\right),
    \end{equation*}
    This downsampling matrix $R$ performs a weighted arithmetic mean over a set of rows of the matrix it is applied on; here, $W \in \mathbb{R}^{8 \times 3}_+$ is downsampled as 
    $RW \in \mathbb{R}^{4 \times 3}_+$. 
    The structure of the matrix $R$ relies on two parameters: $d$ and $f$. The parameter $d$ corresponds to the downsampling ratio. Each row of $R$ has at least $d$ non-zero values that correspond to the rows in $W$ that are combined to form the rows of $RW$; see the underlined entries of $R$ above. 
    The parameter $f$ controls the overlap between the linear combinations of the rows of $W$.   
    In the example above, $f=1$ and one positive value is added to the left and the right end of the $d$ non-zero entries corresponding to the downsampling parameter; 
    see the bold entries in matrix $R$ above.
    These positive values allow an overlap (or coupling) within the downsampling process. If we consider two consecutive frequency bins that result from a downsampling operation, it is reasonable to consider that they share common frequency bins in the original frequency space. We imposed $f \leq \nicefrac{d}{2}$ to avoid too much non-physical coupling. This limitation is also based on numerical experiments that show a degradation of the results when $f$ exceeds $\nicefrac{d}{2}$. 
    When $f=0$, the downsampling matrix $R$ performs a weighted arithmetic mean over $d$ rows without overlapping. 
Note that such downsampling matrices are sparse and nonnegative. 

{When solving~\eqref{MR_betaNMF_model}, we will alternatively update $W$, $H$, and the non-zero entries of $R$ and $S$. As far as we know, this is the first time the matrices $R$ and $S$ are learned simultaneously with the factors $W$ and $H$.}  

\subsection{Scope of this paper} To estimate $V$, $W$ and $H$ from $X$ and $Y$, 
other models \revise{exist}. In particular, for the HSI-MSI fusion, many other approaches have been proposed, e.g.,  
based on tensor decompositions~\cite{Kanatsoulis2018Ma, zhang2018spatial}, 
or based on deep neural networks~\cite{palsson2017multispectral, yang2018hyperspectral}. 
In this paper, we focus on the above linear assumptions and the corresponding coupled NMF model~\eqref{MR_betaNMF_model},  
which have been shown to provide  state-of-the-art results for HSI-MSI fusion; see the survey~\cite{yokoya2017}. {More precisely, we will focus on the use of any  $\beta$-divergence, which has not been before.
As we will see in Section~\ref{sec_audioexp}, $\beta$-divergence for $\beta \neq 2$ allows to obtain improved separation for audio source separation compared to standard NMF algorithms. 
In Section~\ref{sec_numexpimage}, we will show that using the Kullback-Leibler divergence ($\beta=1$) outperforms standard linear models in the presence of Poisson noise.}


\section{Algorithm for MR-$\beta$-NMF}\label{sec_algo}

Most NMF algorithms are based on an iterative scheme that alternatively update $H$ for $W$ fixed and vice versa, and we adopt this approach in this paper. 
The goal in this section is to derive an algorithm to solve  MR-$\beta$-NMF~\eqref{MR_betaNMF_model}.

For $R, S$ and $W$ fixed, let us consider the subproblem in~$H$: 
\begin{equation}\label{subprob_H}
 \underset{H \geq 0}{\text{min}} \;  
 L(H)= D_{\beta}(X\|RWH) + \lambda D_{\beta}(Y\|WHS). 
\end{equation} 
The subproblems in $W$, $R$ and $S$ can be solved similarly.  
To tackle this problem, we follow the standard majorization-minimization (MM) framework \cite{sun2017majorization}. 
{Note however that, because of the sum of the two terms $D_{\beta}(X\|RWH)$ and $D_{\beta}(Y\|WHS)$, the update for $H$ does not follow directly from previous MU derived in the literature. MM algorithms are indeed available for the two terms separately \cite{Fevotte_betadiv}, but a more general auxiliary function needs to be used for the joint problem. The auxiliary function, that we denote $\Bar{L}$, must be a tight upper-bound for the objective $L$ at the current iterate $\tilde{H}$. It is formally defined as follows.}  
\begin{definition}\label{def1}
The function $\bar{L}(H\|\tilde{H}): \Omega \times \Omega \rightarrow \mathbb{R}$ is an auxiliary function for $L\left( H \right):\Omega \rightarrow \mathbb{R}$ at $\tilde{H} \in \Omega$ if the conditions  
$\bar{L}(H\|\tilde{H}) \geq L\left( H \right)$ 
for all $H \in \Omega$ and $\bar{L}(\tilde{H}\|\tilde{H}) = L( \tilde{H} )$ 
are satisfied.
\end{definition}
The optimization problem with $L$ is then replaced by a sequence of simpler problems for which the objective is $\bar{L}$. 
The new iterate $H^{(i+1)}$ is computed by minimizing  the auxiliary function at the previous iterate $H^{(i)}$, either approximately or exactly. 
This guarantees $L$ to decrease at each iteration. 
\begin{lemma}\label{lemmamono}
Let $H, H^{(i)} \geq 0$, and let $\bar{L}$ be an auxiliary function for $L$ at $H^{(i)}$. 
Then $L$ is non-increasing under the update
$H^{(i+1)} = \underset{H \geq 0}{\text{argmin }} \bar{L}(H\|H^{(i)})$. 
\end{lemma}
\begin{proof}
By definition, 
$L(H^{(i)}) = \bar{L}(H^{(i)}\|H^{(i)}) 
\geq 
\underset{H}{\text{min }} \bar{L}(H\|H^{(i)}) 
= \bar{L}(H^{(i+1)}\|H^{(i)})\geq L(H^{(i+1)})$. 
\end{proof}

The most difficult part in using the majorization-minimization framework is to design an auxiliary function that is easy to optimize. Usually such auxiliary functions are separable (that is, there is no interaction between the variables so that each entry of $H$ can be updated independently) and convex. We will construct an auxiliary function for $L(H)$ from \eqref{subprob_H} by a positive linear combination of two auxiliary functions, one for each term of $L(H)$.

\subsubsection{Separable auxiliary function for the first term of $L(H)$}\label{subsec_aux1}

The function $D_{\beta}(X\|RWH)$ separates into $\sum_{n}D_{\beta}(x_{n}\|RWh_{n})$, where $x_{n}$ and $h_{n}$ are the $n$th column of $X$ and $H$ respectively. 
Therefore we only consider the optimization over one specific column $x$ of $X$ and $h$ of $H$. To simplify notation, we denote the current iterate as $\tilde{h}$. 
We now use the separable auxiliary function presented in \cite{Fevotte_betadiv} which consists in majorizing the convex part of the $\beta$-divergence using Jensen's inequality and majorizing the concave part by its tangent (first-order Taylor approximation). 
The $\beta$-divergence can be expressed as the sum of a convex, concave, and constant part, such that:
  \begin{equation*}
  d_{\beta}(x\|y)= \check{d}_{\beta}(x\|y)+\hat{d}_{\beta}(x\|y)+\bar{d}_{\beta}(x\|y), 
  \end{equation*}
where $\check{d}$ is convex function of $y$, $\hat{d}$ is a concave function of $y$ and $\bar{d}$ is a constant of $y$, see \cite{Fevotte_betadiv} for the definition of these terms for different values of $\beta$.

By denoting $RW$ by $P$ and $RW\tilde{h}$ by $\tilde{x}$ with entries $\left[RW\tilde{h}\right]_{f}=\tilde{x}_{f}$ for $f \in \left[1,F_{X} \right]$, the auxiliary function for  $\sum_{f}d_{\beta}(x_{f}\|\left[Ph\right]_{f})$ at $\tilde{h}$ is given by:
\begin{equation}\label{auxfun_term1}
\begin{aligned}
G_{X}(h\|\tilde{h})&=\sum_{f}^{F_{X}}\left[\sum_{k} \frac{p_{fk} \tilde{h_{k}}}{\tilde{x}_{f}}\check{d}_{\beta}(x_{f}\|\tilde{x}_{f}\frac{h_{k}}{\tilde{h_{k}}})\right] +\bar{d}_{\beta}(x_{f}\|\tilde{x}_{f})  \\
&+\left[\hat{d}^{'}_{\beta}(x_{f}\|\tilde{x}_{f})\sum_{k}p_{fk}(h_{k}-\tilde{h_{k}})+\hat{d}_{\beta}(x_{f}\|\tilde{x}_{f}) \right] . 
\end{aligned}
\end{equation}  
Therefore the function 
\begin{equation}\label{auxfun_term1fin}
G_{X}(H\|\tilde{H})=\sum_{n}G_{X}(h_{n}\|\tilde{h}_{n})
\end{equation} 
is an auxiliary function (convex and separable) for $D_{\beta}(X\|RWH)$ at $\tilde{H}$ where $G_{X}(h\|\tilde{h})$ is given by \eqref{auxfun_term1}.

\subsubsection{Separable auxiliary function for the second term of $L(H)$}\label{subsec_aux2} 

Let $\tilde{y}_{fn}=\left[ WHS\right]_{fn}$ and let us use a result from \cite{Fevotte_betadiv}: 
\begin{equation}\label{auxfun_term2}
\begin{aligned}
G_{Y}(H\|\tilde{H})&=\sum_{f,n}\left[\sum_{k,j} \frac{(w_{fk}s_{jn})\tilde{h}_{kj} }{\tilde{y}_{fn}}\check{d}_{\beta}(y_{fn}\|\tilde{y}_{fn}\frac{h_{kj}}{\tilde{h_{kj}}})\right] \\
&+\bar{d}_{\beta}(y_{fn}\|\tilde{y}_{fn}) +\hat{d}_{\beta}(y_{fn}\|\tilde{y}_{fn})  \\
&+\hat{d}^{'}_{\beta}(y_{fn}\|\tilde{y}_{fn})\sum_{k,j}w_{fk}(h_{kj}-\tilde{h_{kj}})s_{jn} . 
\end{aligned}
\end{equation} 
In \cite{Fevotte_betadiv}, the authors show that \eqref{auxfun_term2} is an auxiliary function (separable and convex) to $D_{\beta}(Y\|WHS)$ at $\tilde{H}$ 
: by construction $G_{Y}(H\|\tilde{H})$ is an upper-bound to $D_{\beta}(Y\|WHS)$ at $\tilde{H}$ and is tight when $H=\tilde{H}$.

\subsubsection{Auxiliary function for multi-resolution $\beta$-NMF}
Based on the auxiliary functions presented in Sections \ref{subsec_aux1} and \ref{subsec_aux2}, we can directly derive a separable auxiliary function $\bar{F}(H\| \tilde{H})$ for multi-resolution $\beta$-NMF \eqref{subprob_H}.

\begin{lemma}\label{lemma2}
For $H\geq 0$, $\lambda >0$, the function 
\begin{equation*}\label{eq:5}
\begin{aligned}
& \bar{L}(H\|\tilde{H}) 
= G_{X}(H\|\tilde{H}) + \lambda G_{Y}(H\|\tilde{H}), 
\end{aligned}
\end{equation*}
where $G_{X}$ is given by \eqref{auxfun_term1fin} and 
$G_{Y}$ by \eqref{auxfun_term2}, 
is a convex and separable auxiliary function for  
$L(H)=D_{\beta}(X\|RWH) + \lambda D_{\beta}(Y\|WHS)$. 
\end{lemma}
\begin{proof} 
This follows directly from \eqref{auxfun_term1fin} and \eqref{auxfun_term2}.
\end{proof}

\subsubsection{Multiplicative updates for MR-$\beta$-NMF}

Given the convexity and the separability of the auxiliary function, the optimum is obtained by canceling the gradient. The derivative of the auxiliary function $\bar{L}(H\|\tilde{H})$ with respect to a specific coefficient $h_{kz}$, with index $z$ identifying the same column specified by $n$ in \eqref{auxfun_term1} and specified by $j$ in \eqref{auxfun_term2}, is given by:
\begin{equation}\label{eq:6}
  \begin{aligned}
  \nabla_{h_{kz}}\bar{L}& = \nabla_{h_{kz}}G_{X}(H\|\tilde{H}) + \lambda \nabla_{h_{kz}} G_{Y}(H\|\tilde{H})\\
  & = \sum_{f}^{F_{X}} p_{fk} \left[\check{d^{'}}_{\beta} \left(x_{fz} \|\tilde{x}_{fz} \frac{h_{kz}}{\tilde{h}_{kz}}  \right)   + \hat{d}^{'}_{\beta}(x_{fz}\|\tilde{x}_{fz}) \right]    \\
  & + \lambda \sum_{f}^{F_{Y}} \sum_{n}^{N_{Y}} w_{fk}s_{zn} [\check{d^{'}}_{\beta} \left(y_{fn} \|\tilde{y}_{fn} \frac{h_{kz}}{\tilde{h}_{kz}}  \right) \\
  & +  \hat{d}^{'}_{\beta}(y_{fn}\|\tilde{y}_{fn}) ].
  \end{aligned}
\end{equation}
For example, for $\beta=1$, \eqref{eq:6} becomes:
\begin{equation}\label{eq:7}
  \begin{aligned}
  \nabla_{h_{kz}}\bar{L}& = \sum_{f}^{F_{X}} p_{fk} \left[1- \frac{x_{fz} \tilde{h}_{kz} \tilde{x}^{-1}_{fz} }{h_{kz}} \right] \\
  & + \lambda \sum_{f}^{F_{Y}} \sum_{n}^{N_{Y}} w_{fk}s_{zn} \left[1- \frac{y_{fn} \tilde{h}_{kz} \tilde{y}^{-1}_{fn} }{h_{kz}} \right].
  \end{aligned}
\end{equation}
\begin{center}
\begin{table*}[ht!]
\begin{center}
\caption{Multiplicative updates 
for MR-$\beta$-NMF~\eqref{MR_betaNMF_model}. }
\label{table:beta-MU}
\begin{tabular}{l}
$\hspace{2cm}  H = \tilde{H} \odot \left( \frac{ \left[ W^{T}\left( R^{T}\left( \left(RW\tilde{H}\right)^{.(\beta-2)}\odot X\right) + \lambda \left( \left(W\tilde{H}S\right)^{.(\beta-2)}\odot Y\right) S^{T}  \right) \right] }{ \left[ W^{T}\left(R^{T}\left(RW\tilde{H}\right)^{.(\beta-1)}+\lambda \left(W\tilde{H}S\right)^{.(\beta-1)}S^{T} \right) \right] }      \right)^{.\gamma(\beta)}$,  \\  
$\hspace{2cm}  W = \tilde{W} \odot \left( \frac{ \left[ \left( R^{T}\left( \left(R\tilde{W}H\right)^{.(\beta-2)}\odot X\right) + \lambda \left( \left(\tilde{W}HS\right)^{.(\beta-2)}\odot Y\right) S^{T}  \right) H^{T} \right] }{ \left[ \left(R^{T}\left(R\tilde{W}H\right)^{.(\beta-1)}+\lambda \left(\tilde{W}HS\right)^{.(\beta-1)}S^{T} \right) H^{T} \right] }      \right)^{.\gamma(\beta)}$,  \\   
$\hspace{2cm}  S = \tilde{S} \odot \left( \frac{\left[ H^{T} \left( W^{T} \left(\left(WH\tilde{S}\right)^{.(\beta-2)}\odot Y \right) \right) \right]}{\left[ H^{T} \left( W^{T} \left(WH\tilde{S}\right)^{.(\beta-1)} \right) \right]} \right)^{.\gamma(\beta)}$,  
$R = \tilde{R} \odot \left( \frac{\left[ \left(  \left(\left(\tilde{R}WH\right)^{.(\beta-2)}\odot X \right)H^{T} \right) W^{T} \right]}{\left[  \left( \left(\tilde{R}WH\right)^{.(\beta-1)} H^{T} \right) W^{T} \right]} \right)^{.\gamma(\beta)}$,  \\ 

\normalsize 
where $A \odot B$ 
(resp.\@ $\nicefrac{\left[ A \right]}{\left[ B \right]}$) is the Hadamard product 
(resp.\@ division) between $A$ and $B$, $A^{(.\alpha)}$ is the element \\ 
\normalsize
-wise $\alpha$ exponent of $A$, $\gamma(\beta)=\frac{1}{2-\beta}$ for $\beta < 1$, $\gamma(\beta)=1$ for $\beta \in \left[1,2\right]$ and $\gamma(\beta)=\frac{1}{\beta-1}$ for $\beta >2$ \cite{Fevotte_betadiv}.  
\normalsize
\end{tabular} 
\end{center}
\end{table*}
\end{center} 
Setting \eqref{eq:7} to zero, we get the following closed-form solution for the  $h_{kz}$ coefficient of $H$: 
\begin{equation}\label{eq:8}
          h_{kz}=\tilde{h}_{kz}\frac{\sum_{f}^{F_{X}} p_{fk} x_{fz} \tilde{x}^{-1}_{fz} + \lambda \sum_{f}^{F_{Y}} \sum_{n}^{N_{Y}} w_{fk}s_{zn} y_{fn} \tilde{y}^{-1}_{fn} }{\sum_{f}^{F_{X}} p_{fk}  + \lambda \sum_{f}^{F_{Y}} \sum_{n}^{N_{Y}} w_{fk}s_{zn} } . 
\end{equation} 
The generalization of the closed-form solution \eqref{eq:8} for any $\beta$ for $H$ is given in Table~\ref{table:beta-MU} in matrix forms. 

{Table~\ref{table:beta-MU} also gives the MU for $W$, $R$ and $S$. 
They are obtained exactly in the same was as for $H$. 
For the update of $S$ that should minimize $D_{\beta}(Y \| WHS)$, use the update of $H$ for the term $D_{\beta}(X \| RWH)$  (that is, taking $\lambda = 0$) where 
$X$ is replaced by $Y$, 
$R$ by $W$, $W$ by $H$, and $H$ by $S$. 
For the update of $W$, use the invariance of \eqref{MR_betaNMF_model} by transposition, that is, 
\[
D_{\beta}(X \| RWH) = D_{\beta}(X^\top \| H^\top W^\top R^\top ) 
\] 
and 
\[ 
 D_{\beta}(Y\|WHS) 
= 
 D_{\beta}(Y^\top\|  S^\top H^\top W^\top). 
\] 
For the update of $R$ that should minimize 
$D_{\beta}(X \| R W H)$, use the update of $H$ for the term $D_{\beta}(X^\top \| H^\top W^\top R^\top )$  (that is, taking $\lambda = 0$) where $X$ 
is replaced by $X^\top$, 
$R$ by $H^\top$, $W$ by $W^\top$, and $H$ by $R^\top$. 
}

\begin{theorem}
The updates provided in Table~\ref{table:beta-MU} are guaranteed to decrease the objective function of~\eqref{MR_betaNMF_model}.  
\end{theorem} 
\begin{proof}
This follows from Lemmas~\ref{lemmamono} and~\ref{lemma2}, and from the derivations above so that Table~\ref{table:beta-MU} provides the closed-form update of the auxiliary function of Lemma~\ref{lemma2}.  
\end{proof}

Algorithm~\ref{MRbetaNMF} summarizes our method to tackle \eqref{MR_betaNMF_model} which, for simplicity, will be referred to as MR-$\beta$-NMF. 
It consists in two optimization loops: 

\begin{itemize}
    \item Loop 1: $W$ and $H$ are alternatively updated with downsampling matrices $R$ and $S$ kept fixed to obtain good estimates for $W$ and $H$. The updates are performed for a maximum number of iterations, MAXITERL1. 
    \item Loop 2: $W$, $H$, $S$ and $R$ are alternatively updated so that the algorithm learns the downsampling matrices. The maximum number of iterations for loop 2 is MAXITERL2. 
\end{itemize}
For the HSI-MSI fusion problem, the matrices $R$ and $S$ are usually known and therefore the parameter MAXITERL2 is set to zero. 
In this paper, the second optimization loop is considered only for the audio spectral unmixing application since the matrices $R$ and $S$ are unknown; see Section~\ref{sec_numexpaudio}. 


After $W$ and $H$ are updated, we normalize $W$ such that $\|W(:,k)\|_{1}=1$ for all $k$, and we normalize $H$ accordingly so that $WH$ remains unchanged. This normalization is commonly used for NMF-based methods and is mainly performed to remove the scaling degree of freedom. 
As a convergence condition, we consider the relative change ratio of the cost function $L$ from \eqref{MR_betaNMF_model}, namely 
$|L^{i}-L^{i+1}| \leq \kappa L^{i}$ where $\kappa$ is a given threshold in $(0,1)$, 
and $i$ is the iteration counter. 
We also stop the optimization process if the number of iterations exceeds the predefined maximum number of iterations.

\algsetup{indent=2em}
\begin{algorithm}[ht!]
\caption{Multiplicative updates for MR-$\beta$-NMF \label{MRbetaNMF}}
\begin{algorithmic}[1] 
\REQUIRE A matrix $X \in \mathbb{R}_{+}^{F_{X} \times N_{X}}$, a matrix $Y \in \mathbb{R}_{+}^{F_{Y} \times N_{Y}}$, an initialization $H \in \mathbb{R}^{K \times N_{X}}_+$, an initialization $W  \in \mathbb{R}_{+}^{F_{Y} \times K}$, a matrix $R \in \mathbb{R}_{+}^{F_{X} \times F_{Y}}$, a matrix $S \in \mathbb{R}_{+}^{N_{X} \times N_{Y}}$, a factorization rank $K$, a maximum number of iterations MAXITERL1, a maximum number of iterations MAXITERL2, a threshold $0<\kappa \ll 1$, and a weight $\lambda>0$
\ENSURE A rank-$K$ NMF $(W,H)$ of $V \approx WH$ with $W \geq 0$ and $H \geq 0$, and matrices $R$ and $S$ such that $X \approx RWH$ and $Y \approx WHS$. 
    \medskip  
\STATE \emph{\% Loop 1}
\STATE $i \leftarrow 0$, $L^{0} = 1$, $L^{1} = 0$. 
\WHILE{$i <$  MAXITERL1 and $\left| \frac{L^{i}-L^{i+1}}{L^{i}} \right| > \kappa$ }
	\STATE \emph{\% Update of matrices $H$ and $W$}
	\STATE Update $H$ and $W$ sequentially; see Table~\ref{table:beta-MU}
	\STATE Compute the objective function $L^{i+1}$ 
	\STATE $(W,H) \leftarrow \text{normalize}\left( W , H \right)$, $i \leftarrow i+1$
\ENDWHILE
\STATE \emph{\% Loop 2}
\STATE $i \leftarrow 0$
\WHILE{$i <$  MAXITERL2 and $\left| \frac{L^{i}-L^{i+1}}{L^{i}} \right| > \kappa$ }
	\STATE \emph{\% Update of matrices $H, W, S$ and $R$}
	\STATE Update $H, W, S, R$ sequentially; see Table~\ref{table:beta-MU}
	\STATE Compute the objective function $L^{i+1}$  
	\STATE $(W,H) \leftarrow \text{normalize}\left( W , H \right)$, $i \leftarrow i+1$ 
\ENDWHILE

\end{algorithmic}  
\end{algorithm} 

It can be verified that the computational complexity of the MR-$\beta$-NMF is asymptotically equivalent to the standard MU for $\beta$-NMF, that is, it requires $\mathcal{O}\left(FNK \right)$ operations per iteration.  


{
\paragraph*{Choice of $\beta$} In practice, a crucial issue is to choose the data fitting term; in our case the value of $\beta$. This is a non-trivial task, and many papers have addressed this issue. \revise{Without} prior knowledge, a standard approach is to use cross-validation, that is, hide a subset of the entries and compare the performance of the different models to predict the hidden entries; see, e.g., the discussion in \cite[Section 5.1.2]{Gillis2020}. 
} 

\section{Numerical experiments on audio datasets} \label{sec_audioexp}

In this section, we perform numerical experiments to validate the effectiveness of MR-$\beta$-NMF on two synthetic audio datasets. 

\subsection{Experimental setup and evaluation}\label{sec_testsetupaudio}

\subsubsection{Data}\label{dataset_des}
The proposed technique for joint factorization of amplitude audio spectrograms is applied to two synthetic audio samples. A dedicated test procedure is presented in Section \ref{test_setup} 
in order to evaluate the performance of MR-$\beta$-NMF based on quantitative criteria detailed in subsection \ref{sec_perfeval}. 
The first audio sample is the first measure of ``Mary had a little lamb" and composed of three notes; $E_{4}$, $D_{4}$ and $C_{4}$. The signal is 5 seconds long and has a sampling frequency $f_{s}=44100$Hz yielding $T=220500$ samples. 

\begin{figure}[h!]
\centering 
        \includegraphics[width=0.4\textwidth]{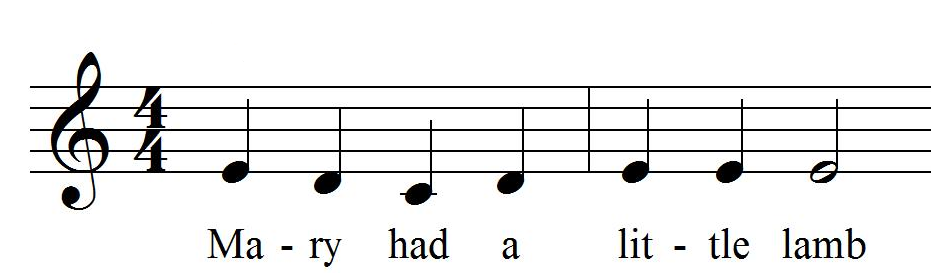}
\caption{
Musical score of ``Mary had a little lamb" (dataset 1).
}
\label{fig:mary_representations}
\end{figure} 

The second audio sample, inspired from \cite{fevotte2009nonnegative}, is a piano sequence played from the score given in Figure \ref{fig:audio2_representations}. The piano sequence is composed of four notes; $D_{4}$, $F_{4}$, $A_{4}$ and $C_{5}$, played all at once in the first measure and then played by pairs in all possible combinations in the remaining measures. The signal is 14.6 seconds long and has a sampling frequency $f_{s}=44100$Hz yielding $T=643817$ samples. 

\begin{figure}[h!]
\centering 
        \includegraphics[width=0.45\textwidth]{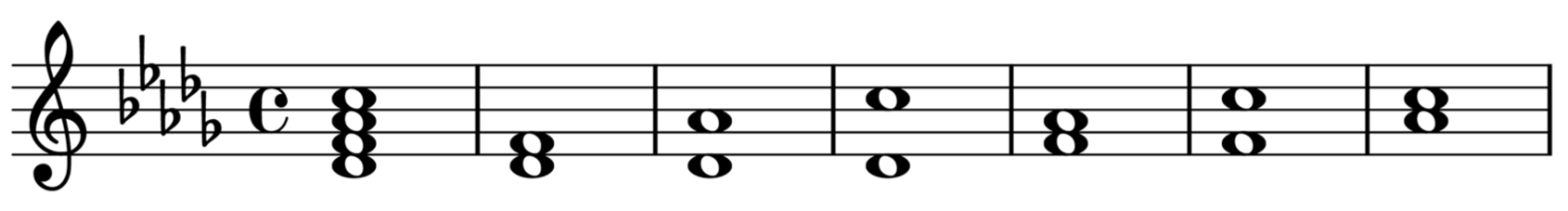}
\caption{
Musical score of the second audio sample (dataset 2).
}
\label{fig:audio2_representations}
\end{figure} 

The music samples have been generated with a professional audio software called Sibelius based on the musical score shown in Figures \ref{fig:mary_representations} and \ref{fig:audio2_representations}. 

\subsubsection{Experimental comparison}\label{test_setup}

This section describes the test procedure elaborated to evaluate the quality of the results obtained with MR-$\beta$-NMF~\eqref{MR_betaNMF_model} that jointly factorizes two audio spectrograms $X$ and $Y$. In the following, matrices $W$ and $H$ stand for the solutions computed with Algorithm~\ref{MRbetaNMF} that solves MR-$\beta$-NMF~\eqref{MR_betaNMF_model}. 
We aim at showing that the factor $W$ has a high frequency resolution whereas the matrix $H$ has a high temporal resolution. To achieve this goal, we compare $W$ to $W_{Y}$ computed with a baseline $\beta$-NMF approach that factorizes the high frequency spectrogram $Y$ only. The baseline $\beta$-NMF applied on $Y$ solves the following optimization problem: 
    \begin{equation}\label{baseline_betaNMF_model_HFR}
    \begin{aligned}
    & \underset{W_{Y}\geq 0, H_{Y}\geq 0}{\text{min}}
    & D_{\beta}(Y\|W_{Y}H_{Y}) . 
    \end{aligned}
    \end{equation} 
Due to the trade-off  between the  frequency  and  temporal  resolutions, the activation matrix  $H_{Y}$ shows a low temporal resolution. 
To compare the accuracy of the solutions $W$ and $W_{Y}$, we need to have access to an oracle matrix $W_{\#}$ that is the reference for the comparison. For instance, for the dataset 1, each column of $W_{\#}$ is supposedly the "true" spectral signature of each of the three notes, namely $E_{4}$, $D_{4}$ and $C_{4}$. We estimated $W_{\#}$ as follows:
\begin{itemize}
    \item We synthetically generate three audio signals and each one contains the sequence of one note in particular. 
    \item Based on the three audio signals, we generate three amplitude spectrograms that have high frequency resolution with the same window size as the one used to generate~$Y$. 
    \item For each amplitude spectrogram, we perform a rank-1 NMF. The resulting  $F_Y$-dimensional vectors are concatenated to form the oracle matrix $W_{\#}$.
\end{itemize}
We show the accuracy of $H$ with a similar procedure; $H$ is compared to an activation matrix $H_{X}$ obtained by solving 
\begin{equation}\label{baseline_betaNMF_modelLFR}
    \begin{aligned}
    & \underset{W_{X}\geq 0, H_{X}\geq 0}{\text{min}}
    & D_{\beta}(X\|W_{X}H_{X}), 
    \end{aligned}
    \end{equation} 
    using multiplicative updates. 
The oracle matrix $H_{\#}$, that is, the reference for the comparison, is computed by performing three independent rank-1 NMF on three amplitude spectrograms that have high temporal resolution, all generated with the same window size as the one used to generate $X$.

\subsubsection{Performance Evaluation}\label{sec_perfeval}
This section presents the qualitative criteria  for  evaluating  the  performance  of the solutions obtained with Algorithm \ref{MRbetaNMF}. We compute the following measures of reconstruction. \\ 
    \noindent $\bullet$ \textit{Activation matrices}: in order to avoid the scaling and permutation ambiguities inherent to the considered NMF models, 
    we first normalize in L-1 norm the rows of the \revise{activation} matrices $H$ and solve an assignment problem w.r.t.\ the oracle matrix $H_{\#}$.
    The quality of the activation matrix $H$ is compared to $H_{X}$ w.r.t.\ $H_{\#}$ by computing the following signal-to-noise ratios (SNR): for all~$k$, 
    \begin{equation}\label{error_activations_MRHR}
    \begin{aligned}
    SNR_{H_{k}}=20 \log_{10} \left( \frac{\|\Bar{H}(k,:)\|_{F}}{\|\Bar{H}(k,:)-\Bar{H}_{\#}(k,:)\|_{F}}\right), 
    \end{aligned}
    \end{equation} 
    where $\Bar{H}(k,:)=\frac{H(k,:)}{\|H(k,:) \|_{1}}$ and $\|H(k,:)\|_{1}=\sum_{j}|H(k,j)|$, and 
    \begin{equation}\label{error_activations_LFRHR}
    \begin{aligned}
    SNR_{H_{X,k}}=20 \log_{10} \left( \frac{\|\Bar{H}_{X}(k,:)\|_{F}}{\|\Bar{H}_{X}(k,:)-\Bar{H}_{\#}(k,:)\|_{F}}\right). 
    \end{aligned}
    \end{equation} 
    The higher the SNRs \eqref{error_activations_MRHR} and \eqref{error_activations_LFRHR}, the better is the estimation for the activation matrix. \\
    \noindent $\bullet$ \textit{Source  matrices}:
    The quality of the source matrix $W$ is evaluated in the same fashion, except that the normalization is performed by columns.  

\subsection{Results}\label{sec_numexpaudio}

\begin{center}
\begin{table*}[h!] 
\begin{center}
\caption{Comparison of MR-$\beta$-NMF with baseline $\beta$-NMF in terms of SNR on the activations and the sources with respect to true factors on the dataset 1. The table reports the average, standard deviation and the best SNR over 100 random initializations for $W$ and $H$. 
Bold numbers indicate the highest SNR.}   
\label{tab:dtm}
\begin{tabular}{|c|cccc|cccc|} 
\hline
Note &  \multicolumn{4}{|c|}{Activation SNRs (dB)} &  
\multicolumn{4}{|c|}{Basis SNRs (dB)}
 \\
 &  \multicolumn{2}{|c|}{$SNR_{H_{k}}$} & \multicolumn{2}{|c|}{$SNR_{H_{X,k}}$} & \multicolumn{2}{|c|}{$SNR_{W_{k}}$} & \multicolumn{2}{|c|}{$SNR_{W_{Y,k}}$} \\
 & average $\pm$ std & best & average $\pm$ std & best & average $\pm$ std & best & average $\pm$ std & best \\
 \hline 
$C_{4}$ & 12.33 $\pm$ 0.17 & \textbf{12.74} & 3.89  $\pm$ 8.99  &  12.19 & 21.35 $\pm$ 1.77 & \textbf{22.66} & 7.95 $\pm$ 7.84 & 12.38\\ 
$D_{4}$ & 14.50 $\pm$ 0.08 & \textbf{14.62} & 8.57  $\pm$ 6.44  &  14.38 & 21.25 $\pm$ 0.35 & \textbf{21.61} & 14.71 $\pm$ 6.06 & 18.23\\ 
$E_{4}$ & 19.68 $\pm$ 0.04 & \textbf{19.82} & 15.28  $\pm$ 5.06   &  19.74 & 22.71 $\pm$ 0.36 & \textbf{23.02} & 19.36 $\pm$ 2.02 & 20.66 \\ 
 \hline  
\end{tabular}
\end{center}
\end{table*}
\end{center} 
\begin{center}
\begin{table*}[h!] 
\begin{center}
\caption{Comparison of MR-$\beta$-NMF with baseline $\beta$-NMF in terms of SNR on the activations and the sources with respect to true factors on the dataset 2. The table reports the average, standard deviation and the best SNR over 100 random \revise{initializations} for $W$ and $H$. 
Bold numbers indicate the highest SNR.}   
\label{tab:dtm_dataset2}
\begin{tabular}{|c|cccc|cccc|} 
\hline
Note &  \multicolumn{4}{|c|}{Activation SNRs (dB)} &  
\multicolumn{4}{|c|}{Sources  SNRs (dB)}
 \\
 &  \multicolumn{2}{|c|}{$SNR_{H_{k}}$} & \multicolumn{2}{|c|}{$SNR_{H_{X,k}}$} & \multicolumn{2}{|c|}{$SNR_{W_{k}}$} & \multicolumn{2}{|c|}{$SNR_{W_{Y,k}}$} \\
 & average $\pm$ std & best & average $\pm$ std & best & average $\pm$ std & best & average $\pm$ std & best \\
 \hline 
$A_{4}$ & 11.98 $\pm$ 0.01 & 12.03 & 12.17  $\pm$ 0.01  &  \textbf{12.17} & 16.24 $\pm$ 0.02 & \textbf{16.43} & 16.29 $\pm$ 0.26 & 16.42\\ 
$C_{5}$ & 9.54 $\pm$ 0.02 & \textbf{9.57} & 9.43  $\pm$ 0.01  &  9.43 & 9.41 $\pm$ 0.02 & \textbf{9.42} & 8.61 $\pm$ 0.72 & 8.73\\ 
$D_{4}$ & 14.81 $\pm$ 0.01 & 14.82 & 14.92  $\pm$ 0.01   & \textbf{14.92} & 16.20 $\pm$ 0.06 & \textbf{16.33} & 15.24 $\pm$ 2.37 & 15.64 \\ 
$F_{4}$ & 11.23 $\pm$ 0.01 & 11.32 & 11.52  $\pm$ 0.01   &  \textbf{11.54} & 16.47 $\pm$ 0.05 & 16.50 & 16.76 $\pm$ 0.99 & \textbf{16.93} \\ 
 \hline  
\end{tabular}
\end{center}
\end{table*}
\end{center} 
In this section, we use the following setting: \\
\noindent $\bullet$ 100 random initializations for $W$ and $H$ for each NMF. \\ 
 \noindent $\bullet$ the window lengths are set to 1024 (23ms) and 4096 (93ms), then the downsampling ratio $d$ is equal to 4. For the generation of $R$ and $S$, the parameter $f$ is set to 2. \\ 
    \noindent $\bullet$ $\beta=1$, and we consider the amplitude spectrograms as the input data. \\
        \noindent $\bullet$ we use $\lambda=1$ in all our experiments.

\subsubsection{Dataset 1: "Mary had a little lamb"}

In this section we report the numerical results obtained after the completion of the test set up presented in section \ref{sec_testsetupaudio}, and using  MAXITERL1=100 and MAXITERL2=400 for Algorithm~\ref{MRbetaNMF}. 

Table \ref{tab:dtm} reports the average SNR, the standard deviation and the best SNR computed for the activations and sources obtained with the models described in Section~\ref{test_setup} over the 100 initializations. 
As it can be observed, activations $H$ are slightly better than activations $H_{X}$, and with a significant smaller standard deviation for each note. The results for the recovered sources  are even more conclusive; MR-$\beta$-NMF outperforms baseline NMF~\eqref{baseline_betaNMF_model_HFR} for which the SNR (best case) can be up to two times larger. 
Moreover, the standard deviations of MR-$\beta$-NMF are significantly lower than those obtained with  baseline NMF~\eqref{baseline_betaNMF_model_HFR}. It appears that the second term in the objective function in \eqref{MR_betaNMF_model} acts as a regularizer so that MR-$\beta$-NMF is more robust to different initializations.

Figure~\ref{fig:basis_1} shows the source  matrices $W_{\#}$, $W$, $W_{Y}$ and $W_{X}$. For more clarity, the frequency range is limited to 2 kHz. This limited range includes all the most significant peaks in terms of magnitude. 
We observe that all the frequency peaks are accurately estimated by MR-$\beta$-NMF for each note. Figure~\ref{fig:basis_1} also integrates the source matrix $W_{X}$ to highlight the impact of 
using baseline NMF~\eqref{baseline_betaNMF_modelLFR} that uses a higher temporal resolution.

\begin{figure}[h!]
\centering 
        \includegraphics[width=0.5\textwidth]{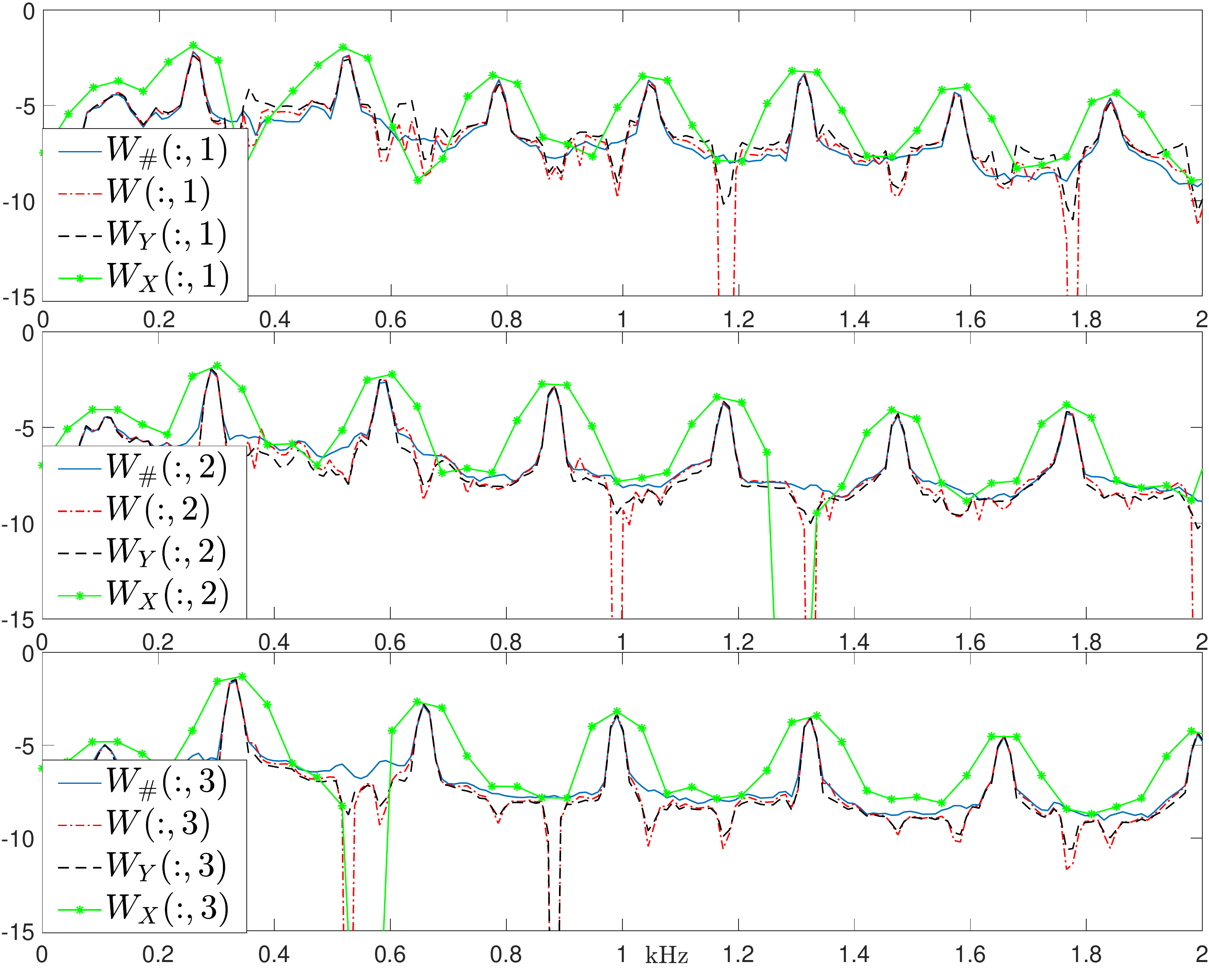}
\caption{
Columns of $W_{\#}$, $W$, $W_{Y}$ and $W_{X}$ in semi-log scale. Top, middle and bottom sub-figures show the spectral content respectively for $C_{4}$, $D_{4}$ and $E_{4}$.
}
\label{fig:basis_1}
\end{figure} 

We conclude that MR-$\beta$-NMF is able to obtain more robust and more accurate results than baseline $\beta$-NMFs that factorize a single spectrogram.

\subsubsection{Data set 2}

In this section we report the numerical results obtained for the dataset 2, using MAXITERL1=500 and MAXITERL2=1500 for Algorithm~\ref{MRbetaNMF}. 

Table \ref{tab:dtm_dataset2} reports the average SNR, the standard deviation and the best SNR computed for activations and sources obtained with the methods described in \ref{test_setup} over 100 initializations. We observe that: \\ 
    \noindent $\bullet$ MR-$\beta$-NMF provides results that show high resolutions in both frequency and temporal domains, \\ 
    \noindent $\bullet$ the regularization effect of  MR-$\beta$-NMF  w.r.t.\ baseline NMFs is less stunning than observed for the dataset 1. However the standard deviations obtained with MR-$\beta$-NMF for the sources  are significantly lower than those obtained with the baseline NMFs. \\ 
    \noindent $\bullet$ by looking more accurately at the results for the sources, MR-$\beta$-NMF globally performs better than  baseline NMFs. For the activations, baseline NMFs perform slightly better than MR-$\beta$-NMF for three scores, with an improvement of at most 1.9\% (for the $F_{4}$ score). 

\section{Numerical experiments on HSI-MSI fusion} \label{sec_hsmsexp}

In this section, we perform numerical experiments to validate the effectiveness of MR-$\beta$-NMF on the HSI-MSI fusion problem.

\subsection{Test setup and criteria }\label{sec_testsetupimage}

\subsubsection{Test data}\label{dataset_desImage}

The proposed MR-$\beta$-NMF algorithm is tested on semi-real datasets against several methods and algorithms widely used to tackle the HSI-MSI fusion problem, namely 
GSA~\cite{Aiazzi2007}, 
CNMF~\cite{Yokoya2012CNMF}, 
HySure~\cite{Simoes2016}, 
FUMI~\cite{wei2016}, 
GLP~\cite{Aiazzi2006}, 
MAPSMM~\cite{Eismann_phd}, 
SFIM~\cite{Liu2000} and 
Lanaras's method~\cite{Lanars2015}. 
In a nutshell: GSA, SFIM and GLP are pansharpening-based methods, the remaining methods belong to subspace-based methods that can be \revise{split} into unmixing methods (CNMF, Lanaras's method and HySure) and Bayesian-based approaches (FUMI, MAPSMM)~\cite{yokoya2017}.

All the algorithms are implemented and tested on a desktop computer with Intel Core i7-8700@3.2GHz CPU, Geforce RTX 2070 Super GPU and 32GB memory. The codes\footnote{\url{https://naotoyokoya.com/Download.html}} are written in MATLAB R2018a. The implementation for benchmarked algorithms comes from the comparative review of the recent literature for HSI-MSI fusion detailed in \cite{yokoya2017}. 
We consider the following real HSI: \\ 
    \noindent $\bullet$ HYDICE Urban: 
    The Urban dataset\footnote{\url{http://lesun.weebly.com/hyperspectral-data-set.html} } consists of 307$\times$307 pixels and 162 spectral reflectance bands in the wavelength range 400nm to 2500nm. We extract a 120$\times$120 subimage from this dataset.  \\  
    \noindent $\bullet$ HYDICE Washington DC Mall: this dataset\footnote{\url{https://engineering.purdue.edu/~biehl/MultiSpec/hyperspectral.html} }\label{perdue} has been acquired with HYDICE HS sensor over the Washington DC Mall and consists of 1208$\times$307 pixels and 191 spectral reflectance bands in the wavelength range 400nm to 2500nm. We extract a 240$\times$240 subimage from this dataset. \\ 
    \noindent $\bullet$ AVIRIS Indian Pines: this dataset has been acquired with NASA Airborne Visible/Infrared Imaging (AVIRIS) Spectrometer \cite{Vane1993} over the Indian Pines test site in North-western Indiana and consists of 145$\times$145 pixels and 200 spectral reflectance bands in the wavelength range 400nm to 2500nm. We extract a 120$\times$120 subimage from this dataset. 

Note that entries of the datasets are uncalibrated relative values, also referred as Digital Numbers (DN). As the goal is to fuse data and not to perform HS unmixing and classification, we do not convert these values into reflectances. 

\subsubsection{Test procedure}\label{testproc_desImage}
In this paper we consider semi-real data by conducting the numerical experiments based on the widely used Wald's protocol \cite{wald2000}. This protocol consists in simulating input MSI and HSI from a reference high-resolution HSI. In this paper, the 
MSI $X$ and HSI $Y$ are obtained from a high-resolution HSI $V$ through the models \eqref{approxV1_final} and \eqref{approxV2_final} respectively. Let us recall that the matrix $R$ from \eqref{approxV1} designates the relative spectral responses from the SR image to the MSI. In other words, it defines how the satellite instruments measure the intensity of the wavelengths (colors) of light. We  generate  a  six-band MSI $X$ by filtering the reference image $V$ with the Landsat 4 TM-like reflectance spectral responses\footnote{\url{https://landsat.usgs.gov/spectral-characteristics-viewer}}. 
The Landsat 4 TM sensor \cite{Chander2009} has a spectral coverage from 400nm to 2500nm so that it is consistent with the spectral coverage of the datasets. 


The matrix $S$ \eqref{approxV2_final} corresponds to the process of spatial blurring and downsampling. The high spectral low spatial resolution HSI $Y$ is generated by applying a 11$\times$11 Gaussian spatial filter with a standard deviation of 1.7 on each band of the reference image $V$ and downsampling every 4 pixels, both horizontally and vertically.
The HSI and MSI are finally both contaminated with noise. The level of noise is usually characterized by the SNR expressed in dB. Here, $\text{SNR}_{X}$ and $\text{SNR}_{Y}$ refer to the noise level for the MSI and HSI,  respectively. In this paper, we apply the same level of noise for each spectral band.
Let us give more insights on the last step of the MS image generation: $X = \max \big( 0 ,  RV + \epsilon_{X} \big)$ where the noise matrix $\epsilon_{X}$ is constructed as follows: we introduce $x_{i}$ for $i=1,2$, some binary coefficients, and 
\[
\tilde{N} = x_1 \frac{N_{\text{P}}}{\|N_{\text{P}}\|_F} 
+ x_2 \frac{N_{\text{F}}}{\|N_{\text{F}}\|_F} , 
\] 
where \\
\noindent $\bullet$ Each entry of $N_{\text{P}}$ is generated using the Poisson distribution of parameter $(R \tilde{V})_{i,j}$ for all $(i,j)$, where $\tilde{V}$ is a noiseless low-rank approximation of $V$ that is computed separately. More precisely, by setting $\epsilon_{X}=0_{F_{X}\times N_{X}}$ where $0_{F_{X} \times N_{X}}$ is all-zero matrix, a solution $\left(W,H \right)$ for MR-$\beta$-NMF  \eqref{MR_betaNMF_model} is first computed with Algorithm~\ref{MRbetaNMF}, and the parameter for the Poisson distribution is defined as $\tilde{V}=WH$. \\  
\noindent $\bullet$ Each entry of $N_{\text{F}}$ is generated using the normal distribution of mean 0 and variance 1. 

We set $\epsilon_{X} = \eta  \frac{\|RV\|_F}{\|\tilde{N}\|_F} \tilde{N}$ with $\eta = \frac{1}{10^{\frac{\text{SNR}_{X}}{20}}}$. For example, if we fix  $\text{SNR}_{X}=25dB$, $V_{1} = \max(0, RV+ \epsilon_{X})$ is a MS image contaminated with  5.62\% of noise (that is, $\|\epsilon_{X}\|_F = 0.0562 \|RV\|_F$) and projected onto the nonnegative orthant. The noise matrix $\epsilon_{Y}$ is obtained in the same way.  


The benchmarked algorithms listed in \ref{dataset_desImage} are configured as recommended in the comparative review \cite{yokoya2017} with the following variations: \\
    \noindent $\bullet$ The number of endmembers is a key parameter for unmixing-based methods. For MR-$\beta$-NMF, CNMF, Lanaras's method and HySure, $K$ is set to the 5 and 6 for HYDICE Urban and HYDICE Washington DC Mall datasets respectively as done in \cite{zhu2017hyperspectral}. For the Indian Pine dataset, $K=16$ as in~\cite{6576298}.\\ 
    \noindent $\bullet$ The benchmarked algorithms are stopped when the relative change of the objective function is below $10^{-4}$ or when the number of iterations exceeds 500. For algorithms such as CNMF that include outer and inner loops, we contacted the authors to set up the best balance for the maximum number of inner ($I_{1}$) and outer ($I_{2}$) loop iterations to fairly compare the methods, the following couples of values are considered: $I_{1}=100$ and $I_{2}=5$ and $I_{1}=250$ and $I_{2}=2$. The couple of values that gives the best results for each dataset is considered in section \ref{sec_numexpimage}, that is $I_{1}=100$ and $I_{2}=5$. \\
    \noindent $\bullet$ The matrix $R$ is known for all algorithms that make use of it. For MR-$\beta$-NMF, it means we use MAXITERL1=500 and MAXITERL2=0.

Finally, let us summarize the initialization strategy: \\ 
    \noindent $\bullet$ MR-$\beta$-NMF uses random nonnegative initializations for $W$ and $H$. \\
    \noindent $\bullet$ CNMF starts by unmixing the HSI using VCA \cite{Nascimento2005} to initialize the endmember signatures, \\
    \noindent $\bullet$ SISAL \cite{Bioucas2009} is used to initialize the endmembers for Lanaras's method. 

Four variants of the MR-$\beta$-NMF are considered, namely $\beta=2$, $\beta=\frac{3}{2}$, $\beta=1$ and $\beta=\frac{1}{2}$.
We test the algorithms under a scenario where no noise is added (that is, $\tilde{N}$ = 0),  
and a scenario where noise is added so that the SNRs for the noise terms in $\epsilon_{X}$ and $\epsilon_{Y}$ are $SNR_{X}=25dB$ and $SNR_{Y}=25dB$.

\subsubsection{Performance evaluation}\label{perfeval_desImage}

In order to assess the fusion quantitatively, we use the following five complementary and widely used quality measurements: \\
    \noindent $\bullet$ Peak SNR (PSNR): the PSNR is used to assess the spatial reconstruction quality of each band. It corresponds to the ratio between the maximum power of a signal and the power of residual errors. A larger PSNR value indicates a higher quality of spatial reconstruction. \\
    \noindent $\bullet$  The root-mean-square error (RMSE):  RMSE is a similarity measure between the SR image $V$ and the fused image $\tilde{V}=WH$. 
    The smaller the RMSE is, the better the fusion quality is.  \\ 
    \noindent $\bullet$ Erreur Relative Globale Adimensionnelle de Synth{\`e}se (ERGAS): ERGAS provides a macroscopic statistical measure of the quality of the fused data. More precisely, ERGAS calculates the amount of spectral distortion in the image \cite{wald2000}. The best value is at 0. \\ 
    \noindent $\bullet$ Spectral Angle Mapper (SAM): SAM is used to quantify the spectral information preservation at each pixel. More precisely, SAM determines the spectral distance by computing the angle between two vectors of the estimated and reference spectra. The overall SAM is obtained by averaging the SAMs computed for all image pixels. The smaller the absolute value of SAM is, the better the fusion quality is. \\  
    \noindent $\bullet$ The universal image quality index (UIQI) introduced in \cite{Wang2002}: UIQI evaluates the similarity between two single-band images. It is related to the correlation, luminance distortion, and contrast distortion of the estimated image w.r.t.\ reference image. UIQI indicator is in the range $\left[-1,1\right]$. For multiband images, the overall UIQI is computed by averaging the UIQI computed band by band. The best value for UIQI is at 1.
    

For more details about these quality measurements, we refer the reader to \cite{Loncan2015}  and \cite{wei2016}.

\subsection{Experimental results}\label{sec_numexpimage}

We ran 20 independent trials for each dataset detailed in \ref{dataset_desImage}. The average performance of each algorithm is shown in Tables \ref{tab:quali_measu1}to \ref{tab:quali_measu3}. 
Except for runtimes, MR-$\beta$-NMF generally rank in the fifth first for all the quality measurements. For Urban dataset with noise added, MR-$\beta$-NMF with $\beta=1$, $\beta=\nicefrac{1}{2}$ and $\beta=\nicefrac{3}{2}$ respectively rank first, second and third for all the metrics except for SAM for which CNMF ranks first. For the condition with no noise added, MR-$\beta$-NMF with $\beta=1$, $\beta=\nicefrac{1}{2}$ ranks first and second for all metrics. MR-$\beta$-NMF with $\beta=\nicefrac{3}{2}$, FUMI and HySure give similar results. 
For Washington DC Mall without noise added, MR-$\beta$-NMF with $\beta=1$, $\beta=\nicefrac{1}{2}$ ranks first and second for all metrics. 
For Indian Pines dataset without noise added, MR-$\beta$-NMF with $\beta=1$ ranks second while HySure ranks first. When noise is added, Lanaras's method ranks first while MR-$\beta$-NMF with $\beta=\nicefrac{1}{2}$, $\beta=1$ rank second and third for most criteria. 

In order to give more insights on the performance comparison between algorithms, 
Figure~\ref{fig:sam_MAPS_Urban} displays the SAM maps obtained for one trial for the Urban, Washington DC Mall and Indian Pines datasets. 
Visually, the proposed method performs competitively with other state-of-the-art methods. Indeed, as already observed with the SAM comparison in Tables \ref{tab:quali_measu1} to \ref{tab:quali_measu3}, the variants of MR-$\beta$-NMF show in general lower values for SAM errors across the images. For the Urban dataset, the highest SAM errors obtained with the variants of MR-$\beta$-NMF are less widespread and localized at some specific spots which correspond to the edges of the roofs and trees. This observation makes sense as those regions show more atypic reflectance angles and therefore more non-linear effects in terms of spectral mixture. The same observations apply for the Washington DC Mall dataset with and without  noise added. For the Indian Pines dataset without noise added, HySure and FUMI algorithms show lower SAM errors \revise{across} images, we visually confirm that MR-$\beta$-NMF with $\beta=1,\nicefrac{1}{2},\nicefrac{3}{2}$ rank third to fifth. When the noise is added, Lanaras's method gives the lowest SAM errors and is less widespread, while MR-$\beta$-NMF with $\beta=1,\nicefrac{1}{2},\nicefrac{3}{2}$ appear to provide less accurate estimates than CNMF that visually looks better. 

{
\subsection{Appendix and discussion}

In the Appendix, we provide additional numerical experiments on the widely used Cuprite data set\footnote{This data sets can be retrieved from the AVIRIS NASA site, \url{https://aviris.jpl.nasa.gov/}.}.  
First, we perform the same experiment as for the Indian Pine data set, for which the conclusions are similar, namely: without noise added, MR-$\beta$-NMF with $\beta$ = 1 ranks second while HySure ranks first. When noise is added, Lanaras’s method ranks first while MR-$\beta$-NMF with $\beta$ = 1/2, $\beta$ = 1 rank second and third for most criteria.  

\revise{The reason other methods sometimes perform better than MR-$\beta$-NMF is because the $\beta$-divergences are guaranteed to perform better only when the data follows certain distributions; for example, the Kullback-Leibler divergence ($\beta = 1$) is the maximum likelihood estimator if the data follows a Poisson distribution. 
This explains why, state-of-the-art methods based on the Frobenius norm sometimes perform similarly as our model based on $\beta$-divergences. 
To validate this behavior,} we also perform a new numerical experiment where we add multiplicative Gamma noise. We show that our proposed MU for $\beta=0$, corresponding to the Itakura-Saito (IS) divergence, outperforms by a large margin all other approaches. This is explained by the fact that the IS divergence  corresponds to the maximum likelihood estimator in the presence of multiplicative Gamma noise~\cite{fevotte2009nonnegative}. This shows that using the right data fitting term can significantly improve the performance of the unmixing. 
} 


\begin{figure*} 
    \begin{tabular}{ccc}
     \includegraphics[width=0.44\textwidth]{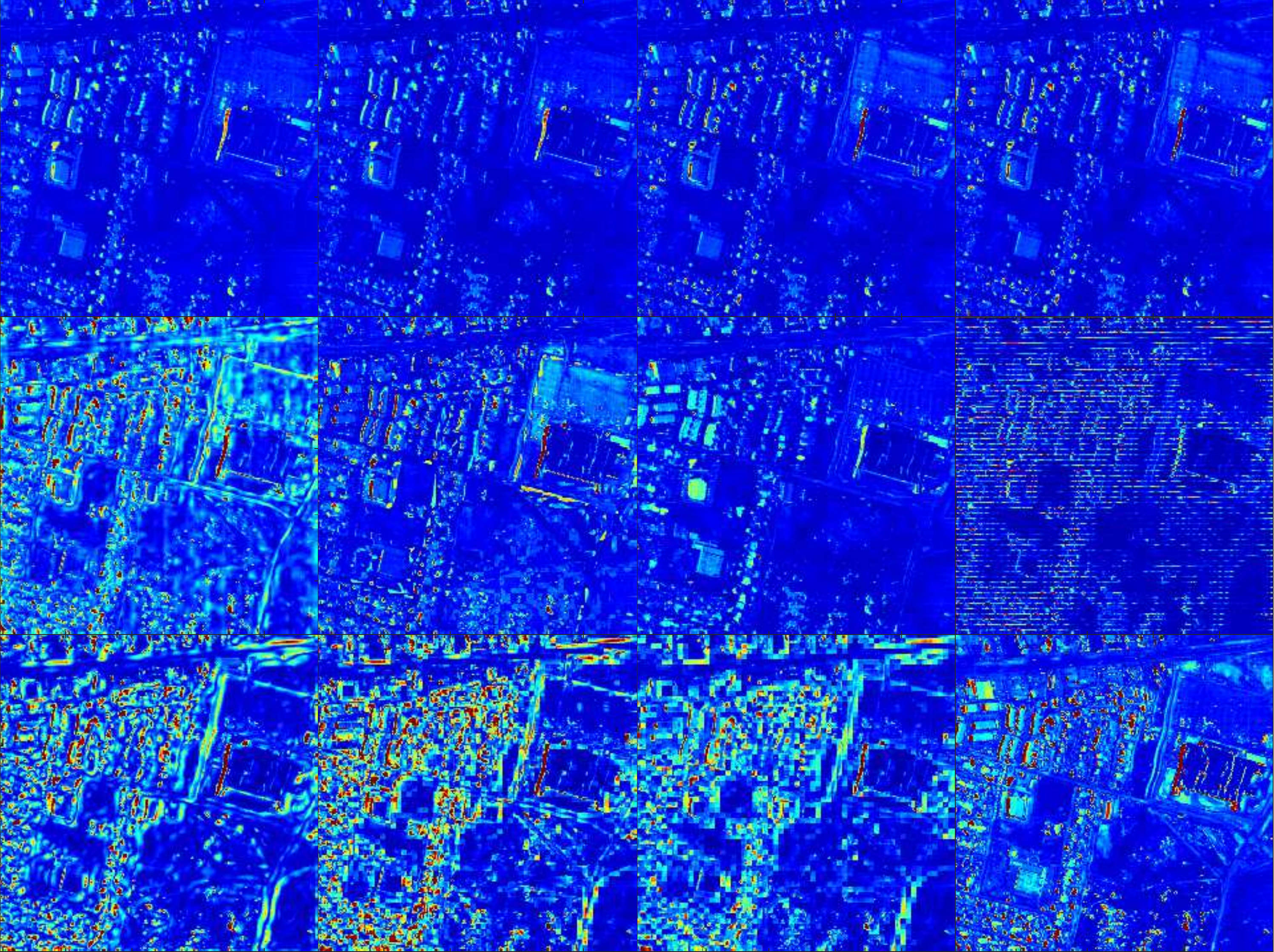}  
     &  
        \includegraphics[width=0.44\textwidth]{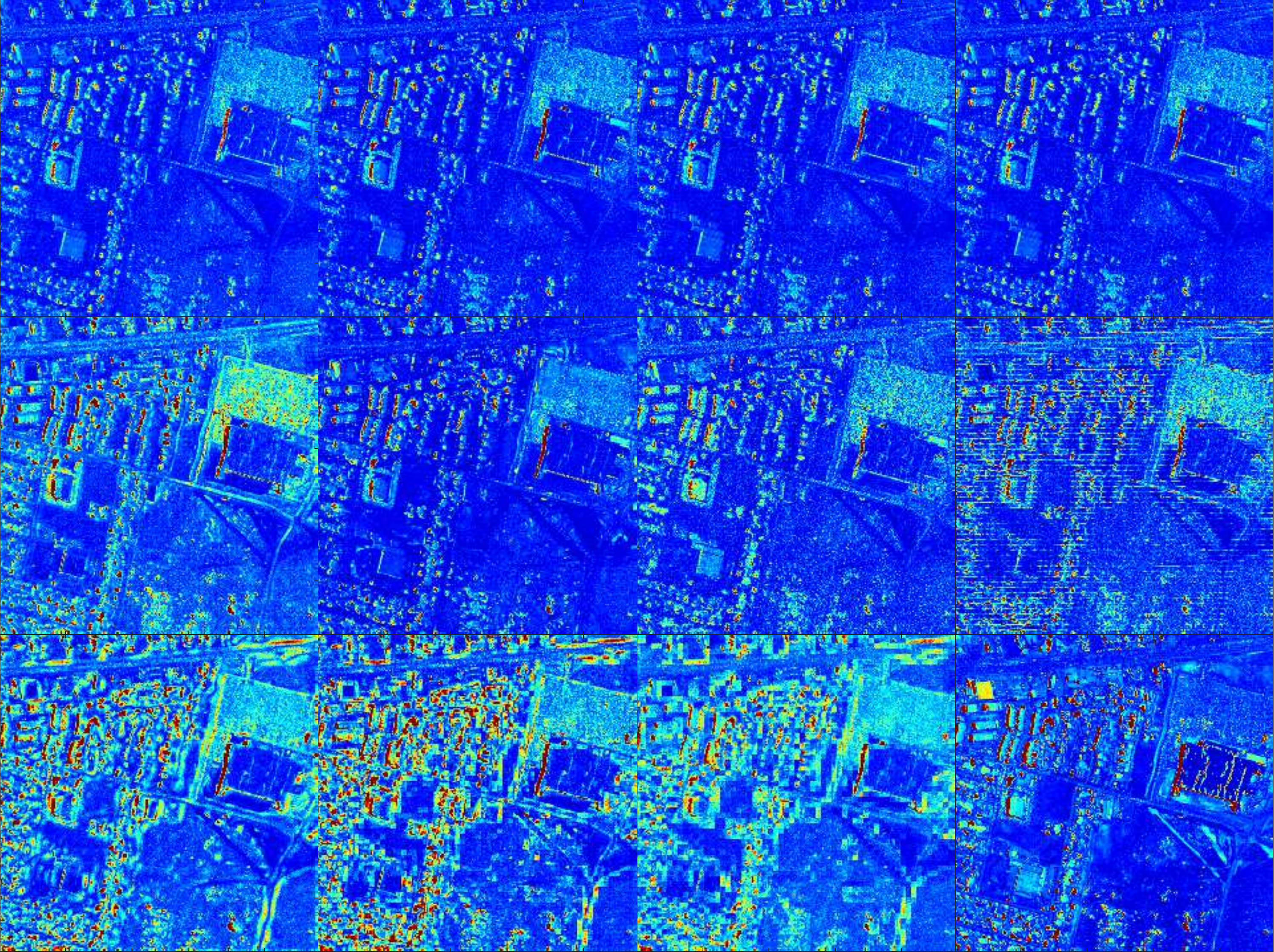}  
        & 
        \includegraphics[width=0.05\textwidth]{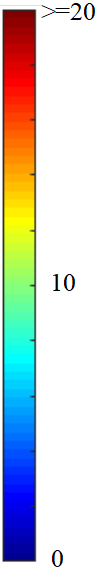} 
        \\ 
        \includegraphics[width=0.44\textwidth]{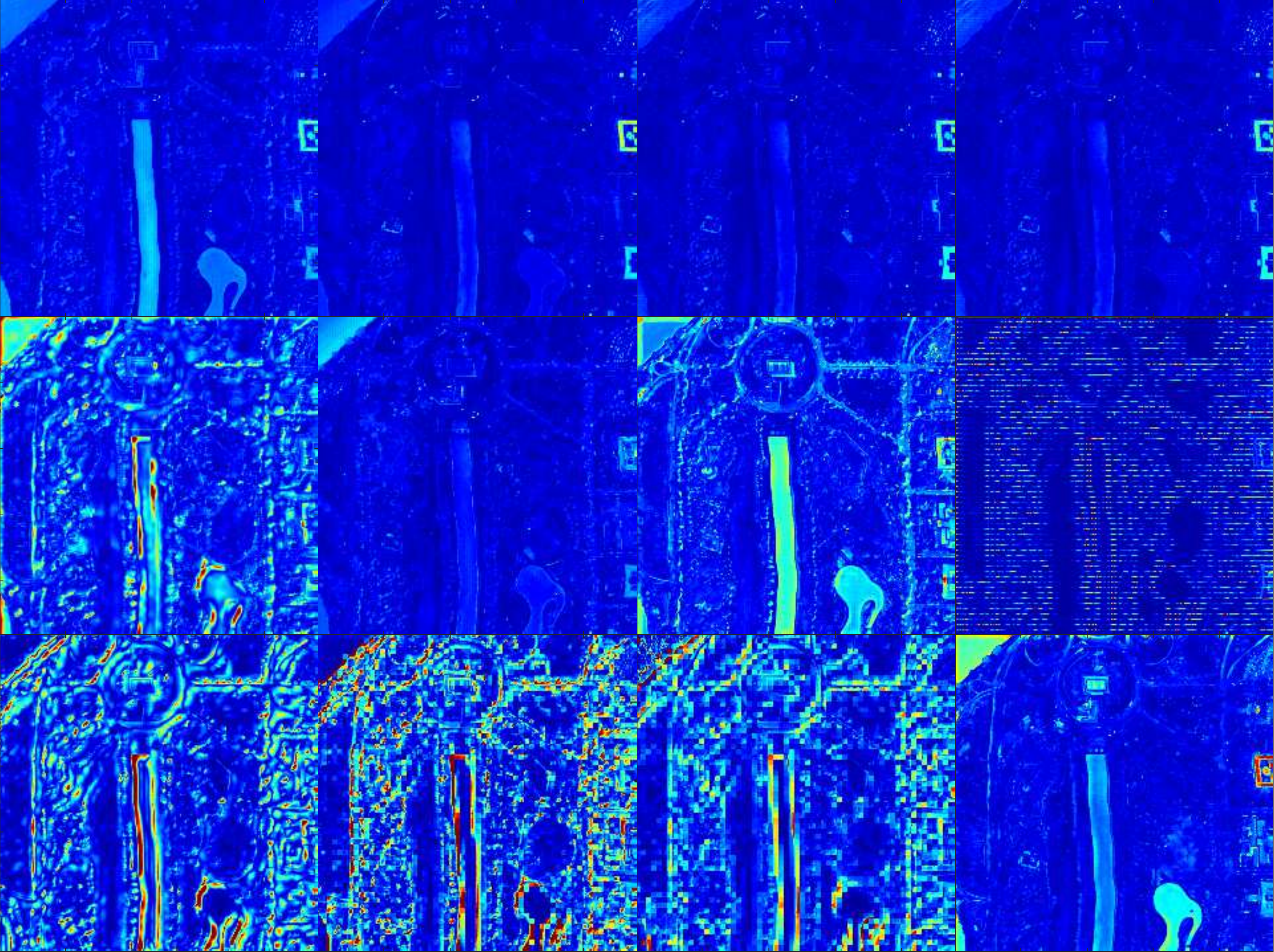}  
        & 
         \includegraphics[width=0.44\textwidth]{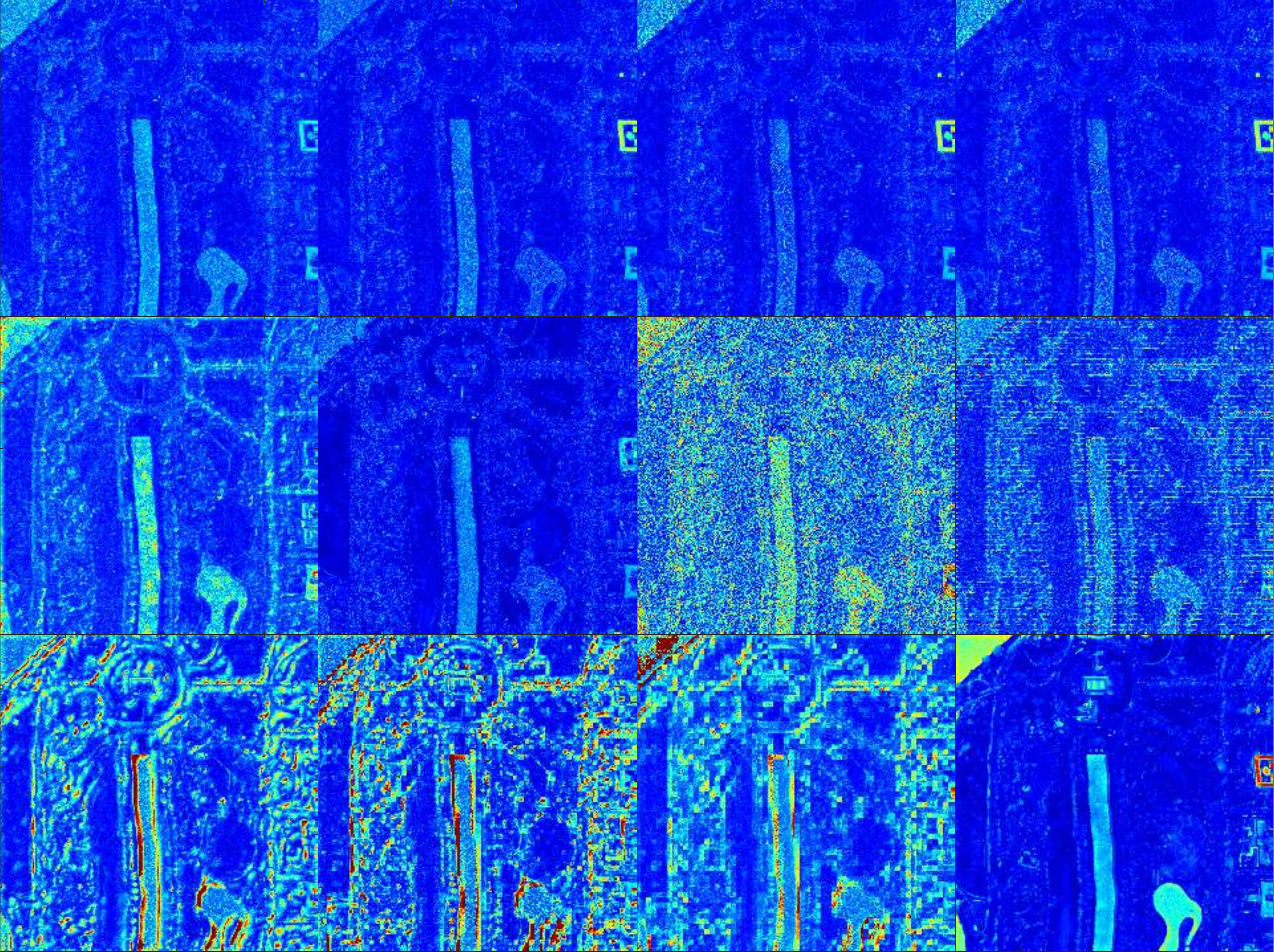} 
         & 
        \includegraphics[width=0.05\textwidth]{Images/axis_sam_MAPS_Urban25dB.png} 
        \\
        \includegraphics[width=0.44\textwidth]{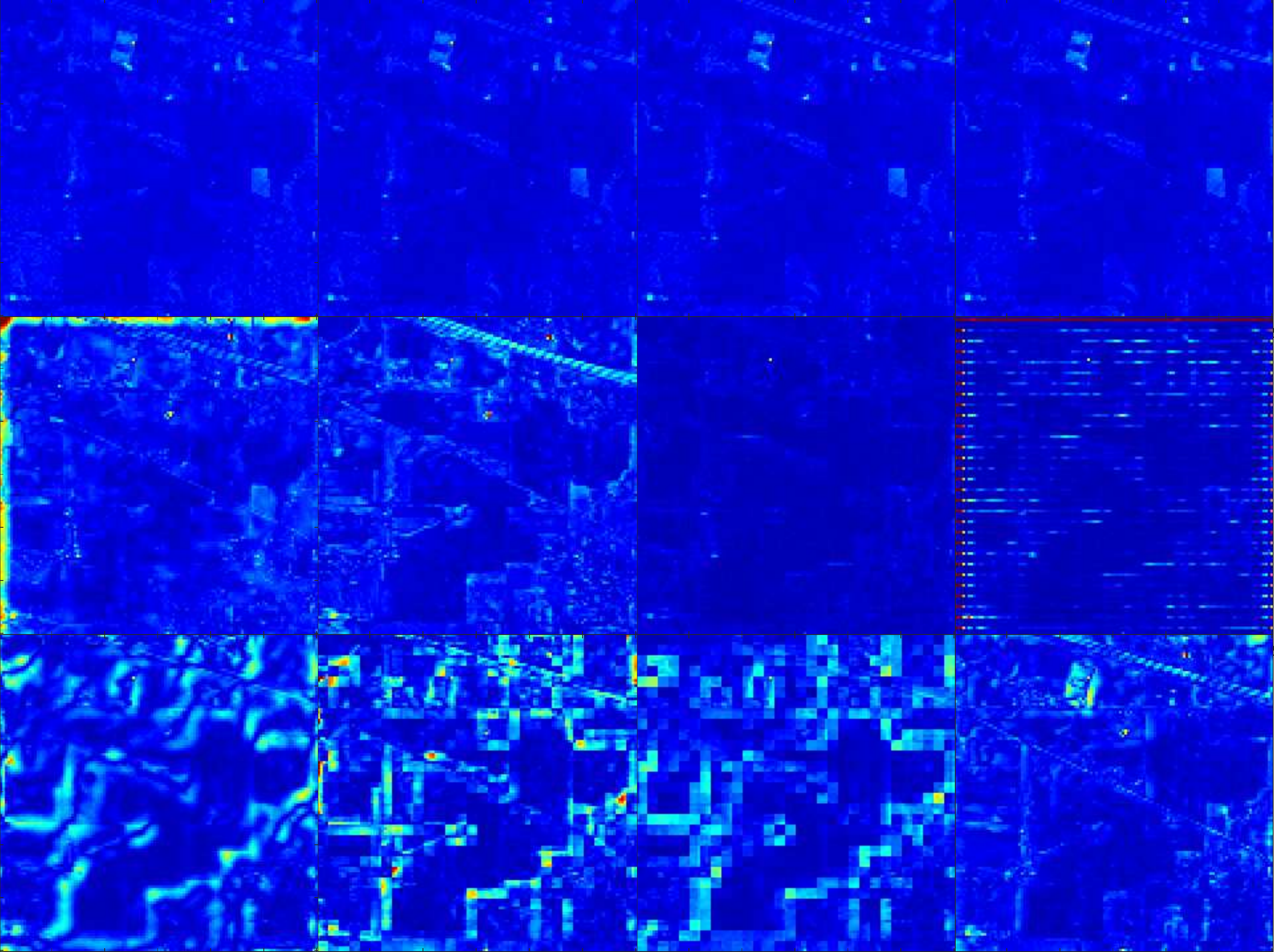} & 
        \includegraphics[width=0.44\textwidth]{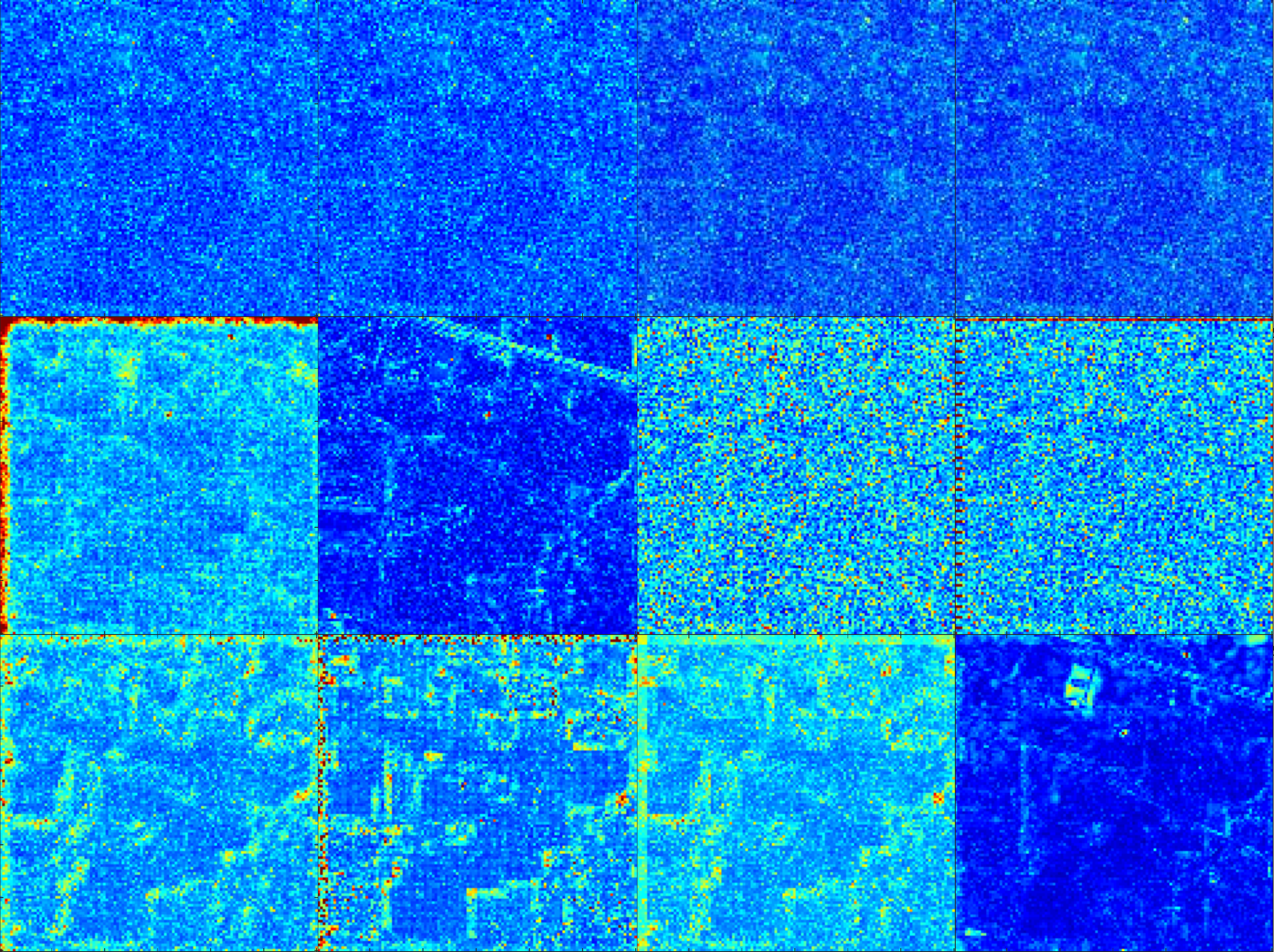} & 
        \includegraphics[width=0.05\textwidth]{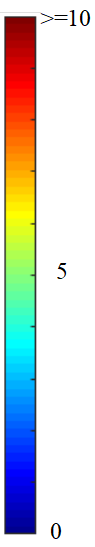}  
    \end{tabular}
    \caption{SAM maps for the different hyperspectral images. 
    From top to bottom: Urban HSI with $K=5$, Washington DC Mall HSI  with $K=6$, and  
    Indian Pines HSI with $K=16$. 
    On the left column: SAM maps without added noise. On the right column: SAM maps with added noise ($SNR_{X}=SNR_{Y}=25dB$). 
    For each image, the 12 SAM maps correspond to the different benchmark algorithms;  from left to right, top to bottom: 
    MR-$2$-NMF, 
    MR-$3/2$-NMF, 
    MR-$1$-NMF, 
    MR-$1/2$-NMF, 
    GSA, 
    CNMF, 
    HySure, 
    FUMI, 
    GLP, 
    MAPSMM, 
    SFIM, and 
    Lanaras's method. }
    \label{fig:sam_MAPS_Urban}
\end{figure*}

\begin{center}
\begin{table*}
\begin{center}
\caption{Comparison of MR-$\beta$-NMF with state-of-the-arts methods for HSI-MSI fusion on the HYDICE Urban dataset.  
The table reports the average, standard deviation for the quantitative quality assessments over 20 trials. Bold, underlined and italic to highlight the three best algorithms.}   
\label{tab:quali_measu1}
\begin{tabular}{|c|c|c|c|c|c|c|} 
\hline
Method      & Runtime (seconds) & PSNR (dB) & RMSE & ERGAS & SAM & UIQI \\
\hline
Best value & 0  & $\infty$ & 0 & 0 & 0 & 1 \\
\hline
\multicolumn{7}{|c|}{Data set - HYDICE Urban - $SNR=25dB$} \\
\hline  
MR-$\beta=2$-NMF & 52.25 $\pm$2.45  & 33.88 $\pm$ 0.10 & 16.26 $\pm$ 0.19 & 2.48 $\pm$ 0.03 & 4.13 $\pm$ 0.06 & 0.97 $\pm$ 0.00 \\
MR-$\beta=3/2$-NMF & 54.46 $\pm$2.31  &\textit{34.54 $\pm$ 0.06} & \textit{14.92 $\pm$ 0.09} & \textit{2.28 $\pm$ 0.01} & 3.65 $\pm$ 0.04 & \textit{0.98 $\pm$ 0.00} \\
MR-$\beta=1$-NMF & 52.20 $\pm$2.03  & \textbf{34.85 $\pm$ 0.10} & \textbf{ 14.51 $\pm$ 0.14} & \textbf{2.22 $\pm$ 0.03} & \textbf{3.49 $\pm$ 0.06} & \textbf{0.98 $\pm$ 0.00} \\
MR-$\beta=1/2$-NMF & 54.47 $\pm$1.96  & \underline{34.81$\pm$ 0.10} & \underline{14.65 $\pm$ 0.15} & \underline{2.24 $\pm$ 0.02} & \underline{3.52 $\pm$ 0.06} & \underline{0.98 $\pm$ 0.00} \\
GSA & \textit{0.72 $\pm$0.05}  & 32.52$\pm$ 0.00 & 19.41 $\pm$ 0.00 & 2.87 $\pm$ 0.00 & 5.63 $\pm$ 0.00 & 0.96 $\pm$ 0.00 \\
CNMF & 9.73 $\pm$1.84  & 34.33$\pm$ 0.50 & 15.45 $\pm$ 0.85 & 2.37 $\pm$ 0.17 & \textit{3.64 $\pm$ 0.27} & 0.98 $\pm$ 0.00 \\
HySure & 31.57 $\pm$2.93  & 33.90$\pm$ 0.00 & 16.44 $\pm$ 0.00 & 2.57 $\pm$ 0.00 & 4.17 $\pm$ 0.00 & 0.97 $\pm$ 0.00 \\
FUMI & \underline{0.39 $\pm$0.03}  & 32.92$\pm$ 0.00 & 20.30 $\pm$ 0.00 & 2.85 $\pm$ 0.00 & 4.92 $\pm$ 0.00 & 0.96 $\pm$ 0.00 \\
GLP & 6.05 $\pm$0.42  & 27.24$\pm$ 0.00 & 34.37 $\pm$ 0.00 & 5.10 $\pm$ 0.00 & 6.27 $\pm$ 0.00 & 0.91 $\pm$ 0.00 \\
MAPSMM & 44.12 $\pm$2.60  & 25.57$\pm$ 0.00 & 41.95 $\pm$ 0.00 & 6.15 $\pm$ 0.00 & 6.82 $\pm$ 0.00 & 0.87 $\pm$ 0.00 \\
SFIM & \textbf{0.24 $\pm$0.03}  & 26.32$\pm$ 0.00 & 37.89 $\pm$ 0.00 & 5.71 $\pm$ 0.00 & 5.90 $\pm$ 0.00 & 0.90 $\pm$ 0.00 \\
Lanaras's method & 8.12 $\pm$8.71  & 29.33$\pm$ 0.29 & 26.84 $\pm$ 0.85 & 4.39 $\pm$ 0.23 & 4.88 $\pm$ 0.26 & 0.94 $\pm$ 0.00 \\
 \hline  
 \multicolumn{7}{|c|}{Data set - HYDICE Urban - No added noise} \\
\hline  
MR-$\beta=2$-NMF & 49.55 $\pm$0.31  & 38.10 $\pm$ 0.40 & 10.94 $\pm$ 0.31 & 1.67 $\pm$ 0.07 & 3.28 $\pm$ 0.10 & 0.99 $\pm$ 0.00 \\
MR-$\beta=3/2$-NMF & 51.54 $\pm$0.52  &40.01 $\pm$ 0.50 & \textit{8.82 $\pm$ 0.32} & \textit{1.35 $\pm$ 0.09} & 2.60 $\pm$ 0.10 & \textit{0.99 $\pm$ 0.00} \\
MR-$\beta=1$-NMF & 49.71 $\pm$0.12  & \underline{41.53 $\pm$ 0.56} & \underline{ 7.86 $\pm$ 0.28} & \textbf{1.19 $\pm$ 0.07} & \underline{2.27 $\pm$ 0.10} & \textbf{0.99 $\pm$ 0.00} \\
MR-$\beta=1/2$-NMF & 52.09 $\pm$0.35  & \textbf{41.69$\pm$ 0.64} & \textbf{7.81 $\pm$ 0.35} & \underline{1.19 $\pm$ 0.08} & \textbf{2.23 $\pm$ 0.12} & \underline{0.99 $\pm$ 0.00} \\
GSA & \textit{0.67 $\pm$0.04}  & 32.93$\pm$ 0.00 & 22.17 $\pm$ 0.00 & 2.87 $\pm$ 0.00 & 5.25 $\pm$ 0.00 & 0.97 $\pm$ 0.00 \\
CNMF & 10.56 $\pm$2.02  & 35.35$\pm$ 0.64 & 13.91 $\pm$ 1.81 & 2.18 $\pm$ 0.32 & 3.26 $\pm$ 0.53 & 0.98 $\pm$ 0.00 \\
HySure & 28.51 $\pm$1.09  & 40.27$\pm$ 0.00 & 9.67 $\pm$ 0.00 & 1.46 $\pm$ 0.00 & \textit{2.50 $\pm$ 0.00} & 0.99 $\pm$ 0.00 \\
FUMI & \underline{0.36 $\pm$0.02}  & \textit{41.01$\pm$} 0.00 & 14.14 $\pm$ 0.00 & 1.67 $\pm$ 0.00 & 2.71 $\pm$ 0.00 & 0.99 $\pm$ 0.00 \\
GLP & 5.61 $\pm$0.09  & 27.97$\pm$ 0.00 & 31.97 $\pm$ 0.00 & 4.65 $\pm$ 0.00 & 4.78 $\pm$ 0.00 & 0.94 $\pm$ 0.00 \\
MAPSMM & 42.19 $\pm$0.84  & 25.92$\pm$ 0.00 & 40.56 $\pm$ 0.00 & 5.89 $\pm$ 0.00 & 5.66 $\pm$ 0.00 & 0.89 $\pm$ 0.00 \\
SFIM & \textbf{0.21 $\pm$0.03}  & 27.05$\pm$ 0.00 & 35.19 $\pm$ 0.00 & 5.21 $\pm$ 0.00 & 4.21 $\pm$ 0.00 & 0.93 $\pm$ 0.00 \\
Lanaras's method & 4.72 $\pm$4.72  & 29.50$\pm$ 0.35 & 26.54 $\pm$ 0.69 & 4.26 $\pm$ 0.23 & 4.57 $\pm$ 0.21 & 0.95 $\pm$ 0.00 \\

\hline  
\end{tabular}
\end{center}
\end{table*}
\end{center}

\begin{center}
\begin{table*}
\begin{center}
\caption{Comparison of MR-$\beta$-NMF with state-of-the-arts methods for HSI-MSI fusion on the HYDICE Washington DC Mall dataset. The table reports the average, standard deviation for the quantitative quality assessments over 20 trials. Bold, underlined and italic to highlight the three best algorithms.}   
\label{tab:quali_measu2}
\begin{tabular}{|c|c|c|c|c|c|c|} 
\hline
Method      & Runtime (seconds) & PSNR (dB) & RMSE & ERGAS & SAM & UIQI \\
\hline
Best value & 0  & $\infty$ & 0 & 0 & 0 & 1 \\
\hline
 \multicolumn{7}{|c|}{Data set - HYDICE Washington DC Mall - $SNR=25dB$} \\
\hline  
MR-$\beta=2$-NMF & 57.59 $\pm$0.32  & \underline{26.77 $\pm$ 0.25} & 202.02 $\pm$ 3.59 & 18.21 $\pm$ 0.13 & 3.38 $\pm$ 0.11 & \textbf{0.90 $\pm$ 0.01} \\
MR-$\beta=3/2$-NMF & 60.04 $\pm$0.39  & \textit{26.37 $\pm$ 0.32} &\textit{ 194.40 $\pm$ 6.38} & \textit{18.07 $\pm$ 0.23} & 3.05 $\pm$ 0.18 & 0.87 $\pm$ 0.01 \\
MR-$\beta=1$-NMF & 57.95 $\pm$0.24  & 26.29 $\pm$ 0.20 &  \textbf{188.42 $\pm$ 11.18} & 18.50 $\pm$ 0.25 & \textit{2.83 $\pm$ 0.28} & 0.86 $\pm$ 0.01 \\
MR-$\beta=1/2$-NMF & 60.38 $\pm$0.20  & 25.68$\pm$ 0.28 & 201.62 $\pm$ 14.05 & 19.46 $\pm$ 0.41 & 3.06 $\pm$ 0.30 & 0.83 $\pm$ 0.01 \\
GSA & \textit{0.79 $\pm$0.04}  & 23.00$\pm$ 0.00 & 235.64 $\pm$ 0.00 & 32.25 $\pm$ 0.00 & 4.20 $\pm$ 0.00 & 0.74 $\pm$ 0.00 \\
CNMF & 7.25 $\pm$1.26  & \textbf{27.60$\pm$ 0.09} & \underline{192.67 $\pm$ 6.50} & \underline{17.37 $\pm$ 0.10} & \textbf{2.55 $\pm$ 0.14} & \textit{0.89 $\pm$ 0.00} \\
HySure & 34.14 $\pm$0.94  & 24.01$\pm$ 0.00 & 351.13 $\pm$ 0.00 & 33.51 $\pm$ 0.00 & 6.15 $\pm$ 0.00 & 0.75 $\pm$ 0.00 \\
FUMI & \underline{0.42 $\pm$0.02}  & 24.67$\pm$ 0.00 & 243.06 $\pm$ 0.00 & 19.73 $\pm$ 0.00 & 4.04 $\pm$ 0.00 & 0.80 $\pm$ 0.00 \\
GLP & 6.42 $\pm$0.24  & 19.85$\pm$ 0.00 & 423.89 $\pm$ 0.00 & 33.64 $\pm$ 0.00 & 5.28 $\pm$ 0.00 & 0.67 $\pm$ 0.00 \\
MAPSMM & 40.91 $\pm$0.46  & 19.34$\pm$ 0.00 & 494.39 $\pm$ 0.00 & 32.18 $\pm$ 0.00 & 5.91 $\pm$ 0.00 & 0.65 $\pm$ 0.00 \\
SFIM & \textbf{0.24 $\pm$0.01}  & 18.08$\pm$ 0.00 & 892.35 $\pm$ 0.00 & 42.23 $\pm$ 0.00 & 5.45 $\pm$ 0.00 & 0.64 $\pm$ 0.00 \\
Lanaras's method & 3.11 $\pm$1.94  & 25.95$\pm$ 0.06 & 235.62 $\pm$ 2.67 & \textbf{17.36 $\pm$ 0.02} & \underline{2.78 $\pm$ 0.03} & \underline{0.90 $\pm$ 0.00} \\
 \hline  
 \multicolumn{7}{|c|}{Data set - HYDICE Washington DC Mall - No added noise} \\
\hline  
MR-$\beta=2$-NMF & 58.55 $\pm$1.50 & 32.61 $\pm$ 0.28 & 128.50 $\pm$ 5.87 & 5.54 $\pm$ 0.13 & 2.59 $\pm$ 0.12 & 0.97 $\pm$ 0.00 \\
MR-$\beta=3/2$-NMF & 60.95 $\pm$1.58  &35.36 $\pm$ 0.38 & \textit{104.11 $\pm$ 5.89} & 2.41 $\pm$ 0.22 & 1.89 $\pm$ 0.12 & \textit{0.98 $\pm$ 0.00} \\
MR-$\beta=1$-NMF & 59.01 $\pm$2.02  & \underline{37.80 $\pm$ 0.75} & \textbf{ 89.20 $\pm$ 5.43} & \underline{1.76 $\pm$ 0.27} & \textbf{1.47 $\pm$ 0.07} & \textbf{0.99 $\pm$ 0.00} \\
MR-$\beta=1/2$-NMF & 61.21 $\pm$1.05  & \textbf{38.27$\pm$ 0.83} & \underline{90.88 $\pm$ 6.26} & \textbf{1.55 $\pm$ 0.20} & \underline{1.48 $\pm$ 0.10} & \underline{0.99 $\pm$ 0.00} \\
GSA & \textit{0.81 $\pm$0.08}  & 29.93$\pm$ 0.00 & 262.27 $\pm$ 0.00 & 3.11 $\pm$ 0.00 & 3.84 $\pm$ 0.00 & 0.97 $\pm$ 0.00 \\
CNMF & 7.90 $\pm$2.67  & 31.46$\pm$ 1.07 & 152.95 $\pm$ 14.25 & 5.93 $\pm$ 8.92 & 2.01 $\pm$ 0.49 & 0.96 $\pm$ 0.03 \\
HySure & 35.85 $\pm$2.19  & 31.23$\pm$ 0.00 & 190.57 $\pm$ 0.10 & 3.21 $\pm$ 0.00 & 3.21 $\pm$ 0.00 & 0.96 $\pm$ 0.00 \\
FUMI & \underline{0.43 $\pm$0.03}  & \textit{36.52$\pm$} 0.00 & 142.92 $\pm$ 0.00 & \textit{2.32 $\pm$ 0.00} & \textit{1.76 $\pm$ 0.00} & 0.98 $\pm$ 0.00 \\
GLP & 6.95 $\pm$0.52  & 26.19$\pm$ 0.00 & 373.07 $\pm$ 0.00 & 4.53 $\pm$ 0.00 & 4.16 $\pm$ 0.00 & 0.93 $\pm$ 0.00 \\
MAPSMM & 42.88 $\pm$0.85  & 24.42$\pm$ 0.00 & 459.09 $\pm$ 0.00 & 5.61 $\pm$ 0.00 & 4.98 $\pm$ 0.00 & 0.88 $\pm$ 0.00 \\
SFIM & \textbf{0.27 $\pm$0.05}  & 25.12$\pm$ 0.00 & 408.40 $\pm$ 0.00 & 6.53 $\pm$ 0.00 & 3.95 $\pm$ 0.00 & 0.92 $\pm$ 0.00 \\
Lanaras's method & 4.70 $\pm$3.55  & 28.46$\pm$ 0.36 & 230.31 $\pm$ 7.44 & 3.94 $\pm$ 0.21 & 2.55 $\pm$ 0.03 & 0.96 $\pm$ 0.00 \\
 \hline  
\end{tabular}
\end{center}
\end{table*}
\end{center}

\begin{center}
\begin{table*}
\begin{center}
\caption{Comparison of MR-$\beta$-NMF with state-of-the-arts methods for HSI-MSI fusion of the dataset AVIRIS Indian Pines dataset. The table reports the average, standard deviation for the quantitative quality assessments over 20 trials. Bold, underlined and italic to highlight the three best algorithms.}   
\label{tab:quali_measu3}
\begin{tabular}{|c|c|c|c|c|c|c|} 
\hline
Method      & Runtime (seconds) & PSNR (dB) & RMSE & ERGAS & SAM & UIQI \\
\hline
Best value & 0  & $\infty$ & 0 & 0 & 0 & 1 \\
\hline
\multicolumn{7}{|c|}{Data set - AVIRIS Indian Pines - $SNR=25dB$} \\
\hline  
MR-$\beta=2$-NMF & 15.48 $\pm$0.53  & 27.11 $\pm$ 0.03 & 187.37 $\pm$ 0.80& 1.64 $\pm$ 0.01 & 2.26 $\pm$ 0.02 & 0.78 $\pm$ 0.00 \\
MR-$\beta=3/2$-NMF & 16.76 $\pm$0.75  &27.29 $\pm$ 0.02 & 183.47 $\pm$ 0.56 & 1.57 $\pm$ 0.00 & 2.14 $\pm$ 0.01 & \textit{0.78 $\pm$ 0.00} \\
MR-$\beta=1$-NMF & 15.57 $\pm$0.53  & \textit{27.38 $\pm$ 0.02} & \textit{181.77 $\pm$ 0.51} & \textit{1.55 $\pm$ 0.00} & 2.09 $\pm$ 0.01 & \underline{0.78 $\pm$ 0.00} \\
MR-$\beta=1/2$-NMF & 16.90 $\pm$0.55  & \underline{27.55$\pm$ 0.03} & \underline{179.10 $\pm$ 0.41} & \underline{1.52 $\pm$ 0.01} & \textit{2.03 $\pm$ 0.01} & \textbf{0.79 $\pm$ 0.00} \\
GSA & \textit{0.31 $\pm$0.04}  & 21.79$\pm$ 0.00 & 326.23 $\pm$ 0.00 & 2.94 $\pm$ 0.00 & 3.28 $\pm$ 0.00 & 0.64 $\pm$ 0.00 \\
CNMF & 2.13 $\pm$0.10  & 24.05$\pm$ 0.21 & 241.72 $\pm$ 5.39 & 2.33 $\pm$ 0.07 & \underline{1.68 $\pm$ 0.04} & 0.60 $\pm$ 0.01 \\
HySure & 22.70 $\pm$0.43  & 24.82$\pm$ 0.28 & 241.17 $\pm$ 3.31 & 2.33$\pm$ 0.13 & 3.25 $\pm$ 0.05 & 0.64 $\pm$ 0.01 \\
FUMI & \textbf{0.12 $\pm$0.02}  & 24.71$\pm$ 0.00 & 242.25 $\pm$ 0.00 & 2.27 $\pm$ 0.00 & 3.19 $\pm$ 0.00 & 0.66 $\pm$ 0.00 \\
GLP & 2.36 $\pm$0.07  & 20.24$\pm$ 0.00 & 403.70 $\pm$ 0.00 & 3.47 $\pm$ 0.00 & 3.14 $\pm$ 0.00 & 0.49 $\pm$ 0.00 \\
MAPSMM & 10.63 $\pm$0.21  & 18.35$\pm$ 0.00 & 519.28 $\pm$ 0.00 & 4.30 $\pm$ 0.00 & 3.36 $\pm$ 0.00 & 0.42 $\pm$ 0.00 \\
SFIM & \underline{0.20 $\pm$0.02}  & 19.74$\pm$ 0.00 & 423.46 $\pm$ 0.00 & 3.68 $\pm$ 0.00 & 3.31$\pm$ 0.00 & 0.48 $\pm$ 0.00 \\
Lanaras's method & 2.82 $\pm$1.69  & \textbf{29.59$\pm$ 0.71} & \textbf{ 149.59 $\pm$ 13.20} & \textbf{1.19 $\pm$ 0.09} & \textbf{1.43 $\pm$ 0.06} & 0.76 $\pm$ 0.05 \\
 \hline  
 \multicolumn{7}{|c|}{Data set - AVIRIS Indian Pines - No added noise} \\
\hline  
MR-$\beta=2$-NMF & 14.55 $\pm$0.07 & 36.43 $\pm$ 0.15 & 69.71 $\pm$ 1.65 & 0.65 $\pm$ 0.02 & 1.23 $\pm$ 0.03 & 0.92 $\pm$ 0.00 \\
MR-$\beta=3/2$-NMF & 15.69$\pm$0.09  &38.09 $\pm$ 0.09 & 57.69 $\pm$ 0.89 & 0.48 $\pm$ 0.00 & 1.00 $\pm$ 0.02 & 0.93 $\pm$ 0.00 \\
MR-$\beta=1$-NMF & 14.56 $\pm$0.03  & \underline{39.30 $\pm$ 0.13} & \underline{ 51.66 $\pm$ 0.79} & \underline{0.41 $\pm$ 0.01} & \underline{0.90 $\pm$ 0.01} & \textit{0.94 $\pm$ 0.00} \\
MR-$\beta=1/2$-NMF & 16.00 $\pm$0.05  & \textit{39.15$\pm$ 0.20} & \textit{52.98 $\pm$ 1.18} & \textit{0.42 $\pm$ 0.01} & 0.91 $\pm$ 0.02 & 0.94$\pm$ 0.00 \\
GSA & \textit{0.29 $\pm$0.03}  & 23.33$\pm$ 0.00 & 300.32 $\pm$ 0.00 & 2.42 $\pm$ 0.00 & 1.38 $\pm$ 0.00 & 0.90$\pm$ 0.00 \\
CNMF & 1.94 $\pm$0.09  & 26.72$\pm$ 0.16 & 184.42 $\pm$ 2.95 & 1.71 $\pm$ 0.04 & 1.17 $\pm$ 0.03 & 0.74 $\pm$ 0.01 \\
HySure & 20.83 $\pm$0.17  & \textbf{40.96$\pm$ 0.03} & \textbf{44.29 $\pm$ 0.18} & \textbf{0.34 $\pm$ 0.00} & \textbf{0.56 $\pm$ 0.00} & \textbf{0.96 $\pm$ 0.00} \\
FUMI & \textbf{0.11 $\pm$0.02}  & 39.13$\pm$ 0.00 & 115.58 $\pm$ 0.00 & 0.83 $\pm$ 0.00 & \textit{0.90 $\pm$ 0.00} & \underline{0.95 $\pm$ 0.00} \\
GLP & 2.24 $\pm$0.05  & 23.12$\pm$ 0.00 & 312.46 $\pm$ 0.00 & 2.48 $\pm$ 0.00 & 1.42 $\pm$ 0.00 & 0.85 $\pm$ 0.00 \\
MAPSMM & 10.09 $\pm$0.14  & 22.27$\pm$ 0.00 & 346.40 $\pm$ 0.00 & 2.74 $\pm$ 0.00 & 1.54 $\pm$ 0.00 & 0.78 $\pm$ 0.00 \\
SFIM & \underline{0.18 $\pm$0.01}  & 22.66$\pm$ 0.00 & 328.92 $\pm$ 0.00 & 2.62 $\pm$ 0.00 & 1.39 $\pm$ 0.00 & 0.85 $\pm$ 0.00 \\
Lanaras's method & 2.05 $\pm$1.90  & 29.89$\pm$ 0.54 & 155.03 $\pm$ 7.39 & 1.15 $\pm$ 0.06 & 1.18 $\pm$ 0.02 & 0.81 $\pm$ 0.00 \\
 \hline  
\end{tabular}
\end{center}
\end{table*}
\end{center}

\vspace{-0.5cm} 

\section{Conclusion}\label{sec_conclusions}

In this paper, we have considered the multi-resolution $\beta$-NMF (MR-$\beta$-NMF) problem~\eqref{MR_betaNMF_model}.  
The estimation of the sources and their activations relies on the minimization of the $\beta$-divergence, a flexible family of measures of fit. MR-$\beta$-NMF addresses the resolution trade-off between two adversarial dimensions  by  fusing  the  information coming  from  multiple  data  with  different  resolutions  in order  to  produce  a  factorization  with  high  resolutions  for all  the  dimensions. 
We have provided multiplicative updates (MU) to tackle the minimization problem. 
We have showcased the efficiency of the MU on two instrumental examples. The first is the audio spectral unmixing for which the frequency-by-time data matrix is computed with the short-time Fourier transform and is the result of a trade-off between the frequency resolution and the temporal resolution. We highlighted the capacity of this model to provide solutions with high frequency and high temporal resolution. 
 MR-$\beta$-NMF was shown to be well suited for audio applications such as transcription problems, 
 and performs in general better than baseline NMF methods. 
The second is blind hyperspectral unmixing for which the wavelength-by-location data matrix is a trade-off between the number of wavelengths measured and the spatial resolution. We demonstrated the efficiency  of MR-$\beta$-NMF to tackle the HSI-MSI fusion problem compared to state-of-the-art methods.

\paragraph*{Acknowledgment} The authors would like to thank Andersen M.S.\ Ang for the help given for the generation of the audio signals used in this paper, and Xavier Siebert for his insightful comments
that helped us improve the paper.

\ifCLASSOPTIONcaptionsoff
  \newpage
\fi



%


\bibliographystyle{IEEEtran}
\bibliography{Bibliography}

\appendix

\section{HSI-MSI fusion for the Cuprite data set}

This appendix contains additional numerical experiments for the fusion of hyperspectral and multispectral images (HSI-MSI fusion) on the widely used Cuprite data set. 

First, we perform the same experiment as in the paper for the Indian Pine data set, namely comparing state-of-the-art algorithms for the noiseless and noisy (Gaussian and Poisson noise) Cuprite hyperspectral image. 

Then, we also perform a new numerical experiment where we add multiplicative Gamma noise. We show that our proposed MU for $\beta=0$, corresponding to the Itakura-Saito (IS) divergence, outperforms all other approaches in this scenario. This is explained by the fact that the IS divergence  corresponds to the maximum likelihood estimator in the presence of multiplicative Gamma noise~\cite{fevotte2009nonnegative}. \\

As done in the numerical experiments of this paper, we ran 20 independent trials for the Cuprite data set. 
Three MSI-MSI fusion analysis are performed. The first two are the same as in this paper, namely (1) without noise, and  
(2) with Gaussian and Poisson noise added to the images. 
Then, we also ass multiplicative Gamma noise (mean equal to 1 and 5\% of standard deviation) applied to the HSI and MSI.
For such noise statistics (multiplicative Gamma distribution), the most suited $\beta$-divergence for the objective function of the optimization problem is  the Itakura-Saïto divergence ($\beta=0$)~\cite{fevotte2009nonnegative}. 
Therefore we also run our proposed algorithm, MR-$\beta$-NMF, with $\beta=0$. 

The average performance of each algorithm is shown in Table \ref{tab:quali_measu_cuprite}. For analysis (1) and (2), the conclusions are similar to the ones observed for Indian Pines dataset, namely: without noise added, MR-$\beta$-NMF with $\beta$ = 1 ranks second while HySure ranks first. When noise is added, Lanaras’s method ranks first while MR-$\beta$-NMF with $\beta$ = 1/2, $\beta$ = 1 rank second and third for most criteria. 

For the case when Gamma noise is added, MR-$\beta$-NMF with $\beta$ = 0 significantly outperforms the others methods, while  MR-$\beta$-NMF with $\beta$ = 1/2 and $\beta$=1 respectively rank second and third. 
This illustrates the importance of using the right data fitting term depending on the noise statistics.

\begin{center}
\begin{table*}
\begin{center}
\caption{Comparison of MR-$\beta$-NMF with state-of-the-arts methods for HSI-MSI fusion of the dataset AVIRIS Cuprite dataset. The table reports the average, standard deviation for the quantitative quality assessments over 20 trials. Bold, underlined and italic to highlight the three best algorithms.}   
\label{tab:quali_measu_cuprite} 
\footnotesize  
\begin{tabular}{|c|c|c|c|c|c|c|} 
\hline
Method      & Runtime (seconds) & PSNR (dB) & RMSE & ERGAS & SAM & UIQI \\
\hline
Best value & 0  & $\infty$ & 0 & 0 & 0 & 1 \\
\hline
\multicolumn{7}{|c|}{Data set - AVIRIS Cuprite - No added noise (1)} \\
\hline  
MR-$\beta=2$-NMF & 13.01 $\pm$0.27 & 37.75 $\pm$ 0.12 & 54.72 $\pm$ 0.88 & 4.21 $\pm$ 0.08 & 1.11 $\pm$ 0.02 & 0.94 $\pm$ 0.00 \\
MR-$\beta=3/2$-NMF & 14.63$\pm$0.16  &38.66 $\pm$ 0.12 & 48.24 $\pm$ 0.61 & 4.19 $\pm$ 0.10 & 0.97 $\pm$ 0.01 & 0.95 $\pm$ 0.00 \\
MR-$\beta=1$-NMF & 13.02 $\pm$0.15  & \underline{39.56 $\pm$ 0.17} & \underline{ 44.08 $\pm$ 0.53} & 4.16 $\pm$ 0.06 & 0.89 $\pm$ 0.01 & \underline{0.96 $\pm$ 0.00} \\
MR-$\beta=1/2$-NMF & 15.36 $\pm$0.24  & \textit{38.98$\pm$ 0.21} & \textit{46.74 $\pm$ 1.25} & 3.79 $\pm$ 0.04 & 0.95 $\pm$ 0.03 & 0.95$\pm$ 0.00 \\
MR-$\beta=0$-NMF & 13.76 $\pm$0.24  & 37.88$\pm$ 0.25 & 53.60 $\pm$ 1.86 & \textbf{2.41 $\pm$ 0.02} & 1.07 $\pm$ 0.04 & 0.94$\pm$ 0.00 \\
GSA & \textbf{0.09 $\pm$0.05}  & 5.10$\pm$ 0.00 & 2694.68 $\pm$ 0.00 & 24.15 $\pm$ 0.00 & 81.44 $\pm$ 0.00 & 0.09$\pm$ 0.00 \\
CNMF & 1.72 $\pm$0.19  & 20.23$\pm$ 0.06 & 375.37 $\pm$ 2.67 & 4.73 $\pm$ 0.02 & 0.70 $\pm$ 0.01 & 0.62 $\pm$ 0.00 \\
HySure & 9.66 $\pm$0.29  & \textbf{41.03$\pm$ 0.05} & \textbf{39.66 $\pm$ 0.40} & \underline{2.68 $\pm$ 0.08} & \textbf{0.60 $\pm$ 0.01} & \textbf{0.97 $\pm$ 0.00} \\
FUMI & \underline{0.11 $\pm$0.03}  & 37.88$\pm$ 0.00 & 102.43 $\pm$ 0.00 & \textit{2.82 $\pm$ 0.00} & 0.86 $\pm$ 0.00 & \textit{0.95 $\pm$ 0.00} \\
GLP & 2.72 $\pm$0.09  & 23.28$\pm$ 0.00 & 265.16 $\pm$ 0.00 & 3.51 $\pm$ 0.00 & \underline{0.65 $\pm$ 0.00} & 0.86 $\pm$ 0.00 \\
MAPSMM & 21.63 $\pm$0.52  & 22.59$\pm$ 0.00 & 287.05 $\pm$ 0.00 & 3.68 $\pm$ 0.00 & \textit{0.65 $\pm$ 0.00} & 0.83 $\pm$ 0.00 \\
SFIM & \textit{0.14 $\pm$0.01}  & 22.80$\pm$ 0.00 & 278.27 $\pm$ 0.00 & 16.82 $\pm$ 0.00 & 0.65 $\pm$ 0.00 & 0.86 $\pm$ 0.00 \\
Lanaras's method & 19.04 $\pm$4.20  & 28.84$\pm$ 0.84 & 137.94 $\pm$ 12.71 & 4.63 $\pm$ 0.31 & 1.39 $\pm$ 0.21 & 0.76 $\pm$ 0.07 \\
 \hline  
\multicolumn{7}{|c|}{Data set - AVIRIS Cuprite - $SNR=25dB$ (2)} \\
\hline  
MR-$\beta=2$-NMF & 12.97 $\pm$0.40  & 26.81 $\pm$ 0.04 & 176.60 $\pm$ 0.91 & 5.21 $\pm$ 0.11 & 2.56 $\pm$ 0.03 & 0.71 $\pm$ 0.00 \\
MR-$\beta=3/2$-NMF & 14.64 $\pm$0.38  &26.79 $\pm$ 0.04 & 176.02 $\pm$ 0.63 & 5.14 $\pm$ 0.12 & 2.53 $\pm$ 0.01 & \textit{0.71 $\pm$ 0.00} \\
MR-$\beta=1$-NMF & 13.06 $\pm$0.42  & \textit{26.84 $\pm$ 0.03} & \textit{175.02 $\pm$ 0.66} & \textit{5.07 $\pm$ 0.12} & 2.49 $\pm$ 0.02 & \underline{0.71 $\pm$ 0.00} \\
MR-$\beta=1/2$-NMF & 15.26 $\pm$0.35  & \underline{26.85$\pm$ 0.05} & \underline{174.83 $\pm$ 0.97} & \textbf{4.79 $\pm$ 0.19} & 2.47 $\pm$ 0.02 & \textbf{0.71 $\pm$ 0.00} \\
MR-$\beta=0$-NMF & 13.65 $\pm$0.33  & 26.63$\pm$ 0.08 & 177.76 $\pm$ 1.67 & 8.67 $\pm$ 0.29 & 2.45 $\pm$ 0.03 & 0.70 $\pm$ 0.01 \\
GSA & \textbf{0.07 $\pm$0.06}  & 3.16$\pm$ 0.00 & 2724.95 $\pm$ 0.00 & 25.53 $\pm$ 0.00 & Inf $\pm$ Inf & 0.00 $\pm$ 0.00 \\
CNMF & 1.36 $\pm$0.16  & 20.05$\pm$ 0.12 & 387.29 $\pm$ 5.69 & 6.42 $\pm$ 2.83 & \underline{1.59 $\pm$ 0.05} & 0.52 $\pm$ 0.02 \\
HySure & 9.99 $\pm$0.14  & 18.83$\pm$ 0.40 & 580.01 $\pm$ 45.86 & 11.26$\pm$ 1.88 & 8.84 $\pm$ 0.52 & 0.29 $\pm$ 0.01 \\
FUMI & \textit{0.14 $\pm$0.16}  & 20.14$\pm$ 0.00 & 460.49 $\pm$ 0.00 & 7.75 $\pm$ 0.00 & 7.23 $\pm$ 0.00 & 0.34 $\pm$ 0.00 \\
GLP & 2.74 $\pm$0.14  & 20.37$\pm$ 0.00 & 365.72 $\pm$ 0.00 & 6.13 $\pm$ 0.00 & 2.67 $\pm$ 0.00 & 0.39 $\pm$ 0.00 \\
MAPSMM & 21.71 $\pm$0.50  & 20.12$\pm$ 0.00 & 378.12 $\pm$ 0.00 & 5.99 $\pm$ 0.00 & \textit{2.39 $\pm$ 0.00} & 0.43 $\pm$ 0.00 \\
SFIM & \underline{0.14 $\pm$0.01}  & 19.90$\pm$ 0.00 & 384.45 $\pm$ 0.00 & 12.55 $\pm$ 0.00 & 2.87$\pm$ 0.00 & 0.37 $\pm$ 0.00 \\
Lanaras's method & 18.74 $\pm$5.10  & \textbf{29.53$\pm$ 0.71} & \textbf{ 127.24 $\pm$ 9.90} & \underline{5.01 $\pm$ 0.23} & \textbf{1.53 $\pm$ 0.16} & 0.71 $\pm$ 0.04 \\
 \hline  
\multicolumn{7}{|c|}{Data set - AVIRIS Cuprite - Multiplicative Gamma noise (3)} \\
\hline  
MR-$\beta=2$-NMF & 13.00 $\pm$0.22  & 29.66 $\pm$ 0.06 & 137.33 $\pm$ 1.08& 4.45 $\pm$ 0.07 & 2.30 $\pm$ 0.02 & 0.71 $\pm$ 0.00 \\
MR-$\beta=3/2$-NMF & 14.75 $\pm$0.62  &29.80 $\pm$ 0.04 & 134.49 $\pm$ 0.68 & 4.41 $\pm$ 0.08 & 2.24 $\pm$ 0.01 & 0.71 $\pm$ 0.00 \\
MR-$\beta=1$-NMF & 13.07 $\pm$0.19  & \textit{29.75 $\pm$ 0.04} & \textit{133.11 $\pm$ 0.69} & 4.39 $\pm$ 0.06 & 2.23 $\pm$ 0.01 & \textit{0.71 $\pm$ 0.00} \\
MR-$\beta=1/2$-NMF & 15.30 $\pm$0.13  & \underline{30.28$\pm$ 0.07} & \underline{125.51 $\pm$ 0.96} & \underline{4.00 $\pm$ 0.07} & 2.08 $\pm$ 0.02 & \underline{0.73 $\pm$ 0.00} \\
MR-$\beta=0$-NMF & 13.69 $\pm$0.12  & \textbf{32.18$\pm$ 0.06} & \textbf{98.65 $\pm$ 0.61} & \textbf{2.55 $\pm$ 0.02} & \textbf{1.49 $\pm$ 0.01} & \textbf{0.81 $\pm$ 0.00} \\
GSA & \textit{0.29 $\pm$0.08}  & 19.34$\pm$ 0.00 & 442.91 $\pm$ 0.00 & 6.20 $\pm$ 0.00 & 6.15 $\pm$ 0.00 & 0.42 $\pm$ 0.00 \\
CNMF & 1.38 $\pm$0.20  & 17.71$\pm$ 0.20 & 505.46 $\pm$ 11.47 & 6.04 $\pm$ 0.10 & \textit{1.93 $\pm$ 0.03} & 0.39 $\pm$ 0.02 \\
HySure & 9.38 $\pm$0.12  & 21.92$\pm$ 0.30 & 434.95 $\pm$ 18.14 & 9.13$\pm$ 1.68 & 7.86 $\pm$ 0.30 & 0.37 $\pm$ 0.01 \\
FUMI & \textbf{0.12 $\pm$0.06}  & 23.31$\pm$ 0.00 & 339.19 $\pm$ 0.00 & 7.48 $\pm$ 0.00 & 6.25 $\pm$ 0.00 & 0.43 $\pm$ 0.00 \\
GLP & 2.71 $\pm$0.13  & 21.26$\pm$ 0.00 & 334.89 $\pm$ 0.00 & 4.41 $\pm$ 0.00 & 3.27 $\pm$ 0.00 & 0.47 $\pm$ 0.00 \\
MAPSMM & 21.73 $\pm$0.48  & 21.47$\pm$ 0.00 & 328.00 $\pm$ 0.00 & \textit{4.18 $\pm$ 0.00} & 2.91 $\pm$ 0.00 & 0.57 $\pm$ 0.00 \\
SFIM & \underline{0.14 $\pm$0.01}  & 20.80$\pm$ 0.00 & 352.83 $\pm$ 0.00 & 5.64 $\pm$ 0.00 & 3.52$\pm$ 0.00 & 0.45 $\pm$ 0.00 \\
Lanaras's method & 24.62 $\pm$3.70  & 28.02$\pm$ 0.56 &  150.84 $\pm$ 9.76 & 4.72 $\pm$ 0.17 & \underline{1.74 $\pm$ 0.12} & 0.69 $\pm$ 0.04 \\
 \hline
\end{tabular}
\end{center}
\end{table*}
\end{center}

%








\end{document}